\documentclass[11pt]{article}

\usepackage[T1]{fontenc}
\usepackage{lmodern}
\usepackage[margin=1in]{geometry}
\usepackage{amsmath,amssymb,amsthm}
\usepackage{mathtools}
\usepackage{bm}
\usepackage{booktabs}
\usepackage{graphicx}
\usepackage{hyperref}
\usepackage{xcolor}
\usepackage{tabularx}
\usepackage{multirow}
\usepackage{longtable}
\usepackage{float}
\usepackage{url}
\usepackage{array}
\usepackage{ragged2e}
\usepackage{algorithm}
\floatstyle{ruled}
\restylefloat{algorithm}
\usepackage{algpseudocode}
\usepackage{microtype}
\hypersetup{hidelinks}

\allowdisplaybreaks
\newcolumntype{Y}{>{\RaggedRight\arraybackslash}X}
\newcolumntype{P}[1]{>{\RaggedRight\arraybackslash}p{#1}}

\DeclareMathOperator{\expit}{expit}
\DeclareMathOperator{\logit}{logit}
\DeclareMathOperator{\E}{E}
\DeclareMathOperator{\Var}{Var}
\DeclareMathOperator{\Prb}{P}

\DeclareMathOperator{\argmax}{arg\,max}

\newtheorem{theorem}{Theorem}
\newtheorem{proposition}{Proposition}

\theoremstyle{definition}
\newtheorem{assumption}{Assumption}
\newtheorem{definition}{Definition}
\theoremstyle{remark}
\newtheorem{remark}{Remark}

\newcommand{\bI}{\bm{I}}
\newcommand{\bJ}{\bm{J}}
\newcommand{\calX}{\mathcal{X}}
\newcommand{\calC}{\mathcal{C}}
\newcommand{\calU}{\mathcal{U}}
\newcommand{\calR}{\mathcal{R}}
\newcommand{\calN}{\mathcal{N}}

\begin{document}

\title{Sequential Certification of Threshold Decisions in Rare-Event Risk
Prediction}
\author{Hui-Mean Foo \qquad Yuan-chin Ivan Chang\thanks{Corresponding author: Yuan-chin Ivan Chang, Institute of Statistical Science, Academia Sinica, Taipei, Taiwan. Email: ycchang@sinica.edu.tw}\\[0.6em]
\small Institute of Statistical Science, Academia Sinica, Taipei, Taiwan}
\date{September 1, 2026}
\maketitle

\begin{abstract}
Rare clinical outcomes pose a difficulty deeper than ordinary class
imbalance: a penalized logistic model can return finite, stable-looking
coefficients before the data support a reliable threshold decision.  We
formulate the accrual question as decision-targeted sequential certification.
On a prespecified finite monitoring schedule and target set, asymptotic
prediction bands are adjusted for simultaneous coverage, and certification is
assessed only among profiles that might be referred, so that a large low-risk
majority cannot trigger an uninformative stop.  Under the working rare-event
logistic model and stated regularity conditions, each certified decision is
asymptotically model-conditionally correct with probability at least
\(1-\alpha\) over the schedule, and the effective information scales with the
number of genuine events.  Simulations and a US linked birth/infant-death
application illustrate the gap between a model that is merely estimable and one
whose decisions are certifiable: the whole-population rule stopped while more
than half of referral-relevant profiles remained ambiguous, whereas the
decision-targeted rule did not certify by the 300{,}000-birth horizon despite
stable temporal validation.  Regularization makes a rare-event model estimable
but does not substitute for genuine rare-event information.
\end{abstract}

\bigskip

\section{Introduction}
\label{sec:intro}

Rare medical outcomes pose a problem more fundamental than ordinary class
imbalance.  When a clinical risk model is updated as patients accrue, only a few
events may have been observed early on; a penalized logistic fit can then return
finite coefficients and apparently stable probabilities while the data still carry
too little genuine event information to support a precise clinical decision.  The
question of interest is therefore not merely whether a classifier discriminates well
at a fixed sample size, but \emph{when enough genuine information has accumulated for
a calibrated risk prediction, and the clinical decision taken from it at a
prespecified threshold, to be made with controlled uncertainty}.  Posed this way,
estimation and classification are one problem: the data yield a fitted risk
\(\widehat p_n(x)\), which yields a decision at a threshold \(\tau\), which is either
certified or deferred until more patients accrue.

Several ingredients of this problem have been studied, but separately.
Many sequential methods have been developed for logistic regression models, and
most target the precision of coefficient estimates through the information matrix,
for example,
fixed-precision sequential inference for logistic regression \cite{changmartinsek1992, grambsch1983}. 
 In addition, rare-event logistic regression and small-event bias
have been studied extensively \cite{kingzeng2001}, and Firth and separation-robust
fitting address the divergence of the ordinary maximum likelihood estimator, though
probability calibration under rare events still requires care
\cite{firth1993,heinzeschemper2002,puhr2017}.  Prediction-model sample-size work
shows that raw events-per-variable rules are inadequate and that calibration and
overall performance matter \cite{vansmeden2019}, and machine-learning risk
prediction for rare outcomes has been pursued in the same clinical spirit
\cite{gruber2020}.  High-dimensional inference for case probabilities is available
under sparsity \cite{guo2021}; confidence-sequence theory supplies genuinely
time-uniform uncertainty for unrestricted monitoring
\cite{howard2021,ramdas2023}, whereas the implementable construction developed here
uses a prespecified finite monitoring grid; and case-control and surrogate-guided
sampling use genuine observations more efficiently under severe imbalance
\cite{fithian2014,tanheagerty2022}.

However, these ingredients have rarely been brought together.  We do not propose
another imbalance correction; rather, we are unaware of a single fixed-precision
formulation that lets rare-event information, regularized risk estimation,
a clinically meaningful threshold, and simultaneous uncertainty over planned
monitoring looks jointly decide when accrual is sufficient.  The clinically relevant
target is a threshold decision over a patient population---not the precision of the
coefficient vector, nor of a single predicted probability at one covariate value.

To supply that formulation, we build on the sequential fixed-accuracy program for
logistic and generalized linear models---fixed-width sequential estimation
\cite{chowrobbins1965,grambsch1983,changmartinsek1992}, stopping-time asymptotics
from nonlinear renewal theory \cite{woodroofe1982,lai1979,chang1996}, and
\(\beta\)-protected sequential estimation \cite{chenwangchang2011}---and repurpose
it for a different object: sequential certification of a clinical threshold
decision under event rarity.  We first show why a natural whole-population stopping
rule can be diluted by the overwhelming low-risk majority.  Our primary proposal is
therefore the decision-targeted rule \(T^{\dagger}\), which controls both ambiguity
and predictive width within the referral-relevant population.  A simultaneous band
over a prespecified finite monitoring schedule ensures that every decision certified
at the selected monitored look is model-conditionally correct on the common coverage
event.  We then derive an oracle information benchmark that combines event rarity,
profile-specific variance, and distance to the threshold, separating a model that is
merely \emph{estimable} from one whose decisions are \emph{certifiable}.  The
classical stopping-time asymptotics and \(\beta\)-protection that we draw on are not
new; the contribution is the decision-targeted certification formulation and its
rare-event interpretation.  A genuinely time-uniform confidence sequence would
extend the method to dense or unbounded monitoring, but that extension is not claimed
for the finite-grid procedure studied here.

The remainder of the main paper follows the decision problem from method to
evidence.  Section~\ref{sec:method} defines the rare-event model, simultaneous
finite-grid band, and decision-targeted stopping rule, and states the central
validity and first-order scaling results.  Section~\ref{sec:num} evaluates the rule
in simulation, Section~\ref{sec:data} gives the temporally validated NCHS
application, and Section~\ref{sec:discussion} discusses interpretation and scope.
Appendices~\ref{app:secondorder}--\ref{app:mody} contain the second-order
stopping-time theory, complete proofs, additional estimator and sensitivity
results, and the MODY-calibrated plasmode study.

\section{Methodology}
\label{sec:method}

\subsection{Model, decision, and predictive band}
\label{sec:setup}

We first fix the four ingredients the rest of the section builds on: a rare-event
logistic model, a single default estimator, the clinical decision the model serves,
and the predictive band from which certification is read.

Separate the asymptotic rarity index from the sequential accrual index.  Let
\(r=1,2,\ldots\) index a sequence of increasingly rare data-generating regimes and
let \(i\) index observations within a regime.  We assume
\begin{equation}\label{eq:model}
Y_{ri}\mid X_{ri}=x_{ri}
 \sim \operatorname{Bernoulli}\{p_r(x_{ri})\},
\qquad
p_r(x)=\expit\{\gamma_r+x^\top\beta_0\},
\qquad
\gamma_r\rightarrow-\infty.
\end{equation}
Writing \(\rho_r=e^{\gamma_r}\rightarrow0\),
\[
p_r(x)=\frac{\rho_r e^{x^\top\beta_0}}
              {1+\rho_r e^{x^\top\beta_0}}
       \sim \rho_r e^{x^\top\beta_0}
\]
for fixed \(x\).  For any given sequential experiment, the regime \(r\), its
intercept \(\gamma_r\), and hence its event distribution are fixed while the accrued
sample size \(n\) increases.  This separation prevents the true event probability
from appearing to change during monitoring.  We suppress the index \(r\) in fitted
quantities when no ambiguity is possible.

Let \(\theta_r=(\gamma_r,\beta_0^\top)^\top\) and
\(\widetilde x=(1,x^\top)^\top\).  Our default implementation is ridge logistic
regression with an unpenalized intercept,
\begin{equation}\label{eq:ridge}
\widehat\theta_{r,n}
=\argmax_{\theta}\left\{\ell_{r,n}(\theta)
-\tfrac12\theta^\top\Lambda_n\theta\right\},
\qquad
\Lambda_n=\operatorname{diag}(0,\lambda_n I_p).
\end{equation}
It yields
\(\widehat\eta_{r,n}(x)=\widetilde x^\top\widehat\theta_{r,n}\) and
\(\widehat p_{r,n}(x)=\expit\{\widehat\eta_{r,n}(x)\}\).  Ridge is used as a stable
default, not as an estimator for which validity is automatic.  The certification
argument below applies to any estimator satisfying the stated prediction-scale
expansion and variance-consistency conditions.  For ridge, a sufficient
small-penalty condition for first-order centering is
\(\lambda_n=o\{\sqrt{n\rho_r}\}\); otherwise its prediction bias must be estimated
and removed or incorporated into the band.

The action the model supports is a threshold decision.  For a prespecified clinical
threshold \(\tau\in(0,1)\), write \(c=\logit(\tau)\); in regime \(r\), the
target decision is
\begin{equation}\label{eq:decision}
\delta_r(x)=I\{p_r(x)\ge\tau\}=I\{\eta_r(x)\ge c\},
\qquad \eta_r(x)=\gamma_r+x^\top\beta_0.
\end{equation}
The threshold reflects clinical utility, treatment burden, or the relative costs of
false positives and negatives, and for a rare outcome a fixed value of \(0.5\) is
rarely appropriate.

Both the precision of \(\widehat p_n\) and the certainty of this decision are read
from a single object, a simultaneous predictive band over a prespecified finite
monitoring schedule.  Let
\([L_n^\eta(x),U_n^\eta(x)]\) cover the linear predictor simultaneously over the
monitored sample sizes and a target set \(\calX_0\),
\begin{equation}\label{eq:coverage}
\Prb_r\left\{\eta_r(x)\in[L_n^\eta(x),U_n^\eta(x)]
\text{ for every monitored }n\text{ and every }x\in\calX_0\right\}\ge 1-\alpha,
\end{equation}
with risk interval \(L_n^p(x)=\expit\{L_n^\eta(x)\}\),
\(U_n^p(x)=\expit\{U_n^\eta(x)\}\).  The same band serves both purposes: its width
measures predictive precision, while its position relative to \(c\) determines
whether a classification is certified.  A concrete construction is given next; for
now it suffices that such a band exists.

\begin{remark}[Alternative penalties]
Firth, Firth--ridge, elastic net, and Bayesian shrinkage are important comparators,
but we do not treat each as a separate methodological branch; their role here is
sensitivity analysis and supplementary comparison.  Firth-type methods are
particularly useful when separation is the dominant concern
\cite{firth1993,heinzeschemper2002,puhr2017}, though such penalization carries its
own caveats for rare-disease regression \cite{hashibe2026}.
\end{remark}

\subsection{Constructing the band}
\label{sec:bandconstruct}

We make the band explicit rather than merely assume it.  Fix a prespecified
monitoring schedule \(\calN=\{n_1<\cdots<n_K\}\) (a geometric grid in the
experiments) and a finite target set \(\calX_0\).

\begin{assumption}[Prediction-scale estimator regularity]
\label{ass:estimator}
Consider a sequence of designs indexed by the rare-event regime \(r\), with
finite monitoring schedules \(\calN_r\) and finite target sets
\(\calX_{0,r}\); subscripts are suppressed when no ambiguity is possible, and
their cardinalities are fixed in the asymptotic argument.  Let
\(\widehat\eta^{\mathrm{bc}}_{r,n}(x)\) denote the fitted linear predictor after any
required first-order penalty-bias correction, and let
\(\widehat V_{r,n}(x)\) consistently estimate its variance.  Uniformly over the
finite set \(\calN_r\times\calX_{0,r}\), as the smallest effective information
\(\min_{n\in\calN_r}n\rho_r\) diverges,
\begin{equation}\label{eq:estregularity}
\frac{\widehat\eta^{\mathrm{bc}}_{r,n}(x)-\eta_r(x)}
     {\sqrt{\widehat V_{r,n}(x)}}
\ \xrightarrow{d}\ N(0,1).
\end{equation}
For the ridge estimator in \eqref{eq:ridge}, one may take
\(\widehat\eta^{\mathrm{bc}}_{r,n}=\widehat\eta_{r,n}\) when
\(\lambda_n=o\{\sqrt{n\rho_r}\}\).  A corresponding penalized model-based covariance
is
\begin{equation}\label{eq:ridgevariance}
\widehat\Sigma_{r,n}
=(\widehat H_{r,n}+\Lambda_n)^{-1}
  \widehat H_{r,n}
 (\widehat H_{r,n}+\Lambda_n)^{-1},
\qquad
\widehat V_{r,n}(x)=\widetilde x^\top
\widehat\Sigma_{r,n}\widetilde x,
\end{equation}
where \(\widehat H_{r,n}\) is the unpenalized observed information.  If the penalty
bias is not negligible, the centering must instead use a justified bias correction
and its associated variance estimate.
\end{assumption}

Under Assumption~\ref{ass:estimator}, take
\begin{equation}\label{eq:band}
[L_n^\eta(x),U_n^\eta(x)]
=\widehat\eta^{\mathrm{bc}}_{r,n}(x)
 \pm a_{\alpha}\sqrt{\widehat V_{r,n}(x)},
\qquad
a_{\alpha}=\Phi^{-1}\!\Bigl(1-\tfrac{\alpha}{2K\,|\calX_0|}\Bigr).
\end{equation}
A union bound over the \(K\,|\calX_0|\) monitored pairs then gives simultaneous
coverage
\[
\Prb_r\bigl\{\eta_r(x)\in[L_n^\eta(x),U_n^\eta(x)]
\ \forall (n,x)\in\calN\times\calX_0\bigr\}
\ \ge\ 1-\alpha
\]
(asymptotically, from the per-look normal approximation).
For a fixed finite collection of monitored pairs, each marginal noncoverage
probability is at most \(\alpha/(K|\calX_0|)+o(1)\), so Bonferroni gives
total noncoverage at most \(\alpha+o(1)\). A growing collection would require
additional control of the sum of the approximation errors; pointwise normality
alone does not supply that control.
\begin{remark}[Finite-grid simultaneous validity]
This is the band used in Study~4 (Section~\ref{sec:sim}), where the naive fixed-\(n\) choice
\(a_\alpha=z_{1-\alpha/2}\) is shown to violate the guarantee under repeated looks
while the corrected \(a_\alpha\) restores it.  The union bound is deliberately
simple and pays a \(\log(K|\calX_0|)\) width penalty; for dense or unbounded
monitoring a genuine time-uniform confidence sequence \cite{howard2021} can replace
it.  The present construction is not claimed to be valid outside the
prespecified finite schedule \(\calN\).
\end{remark}

\begin{remark}[Scope of certification]
\label{rem:modelconditional}
The guarantee below is model-conditional: it certifies the threshold decision
defined by the working logistic risk \(p_r(x)\), provided Assumption~\ref{ass:estimator}
and the simultaneous coverage statement hold.  Under link or linear-predictor
misspecification, coverage of the working-model target does not by itself certify the
true clinical risk.  Out-of-sample calibration and decision performance are therefore
separate empirical requirements, not consequences of the certification theorem.
\end{remark}

Define the certified set
\[
\calC_n
=
\left\{x\in\calX_0:\ U_n^\eta(x)<c\ \text{or}\ L_n^\eta(x)>c\right\},
\qquad
\calU_n=\calX_0\setminus\calC_n.
\]
If \(Q\) denotes the target covariate distribution, define the ambiguity mass
\(A_n=Q(\calU_n)\) and the predictive width functional
\[
P_n=Q_{1-\gamma}\left[U_n^p(X)-L_n^p(X)\right].
\]
Both are computed from the same finite-grid simultaneous band.

\subsection{Whole-population baseline and decision-targeted certification}
\label{sec:decisiontargeted}

{
A natural baseline controls ambiguity and width over the entire target population:
\begin{equation}\label{eq:psi}
\Psi_n^{\mathrm{all}}
=\max\left\{\frac{A_n}{\varepsilon},\frac{P_n}{d}\right\},
\qquad
T^{\mathrm{all}}
=\inf\left\{n\in\calN:\ n\ge n_{\min},\;
\Psi_n^{\mathrm{all}}\le1\right\}.
\end{equation}
Thus \(T^{\mathrm{all}}\) stops when no more than an \(\varepsilon\) fraction of
the entire target population is classification-ambiguous and the
\((1-\gamma)\)-quantile of its risk-interval width is at most \(d\).  We use this
rule as a comparator, not as the primary proposal.
}

\begin{algorithm}[tbp]
\caption{Decision-targeted sequential certification of a rare-event risk model}
\label{alg:main}
\begin{algorithmic}[1]
\Require Initial genuine data \(D_{n_0}\); prespecified monitoring schedule
\(\calN\) and minimum eligible size \(n_{\min}\); target set/distribution
\((\calX_0,Q)\); decision threshold \(\tau\);
tolerances \((\varepsilon,d)\); width quantile parameter \(\gamma\); simultaneous finite-grid level \(1-\alpha\).
\Ensure Decision-targeted stopping time \(T^{\dagger}\), fitted risk model,
and certified decisions.
\State Set \(n\) to the first scheduled look in \(\calN\) not below
\(\max(n_0,n_{\min})\).
\While{a scheduled look remains}
  \State Fit the regularized rare-event logistic model and apply any required
  predictor-bias correction.
  \State Construct the simultaneous finite-grid predictor band
  \([L_n^\eta(x),U_n^\eta(x)]\) on \(\calX_0\).
  \State Form the risk interval \([L_n^p(x),U_n^p(x)]\).
  \State Compute \(\calU_n=\{x:L_n^\eta(x)\le c\le U_n^\eta(x)\}\) and
  \(\calR_n=\{x:U_n^\eta(x)\ge c\}\).
  \If{\(Q(\calR_n)=0\)}
    \State Set \(A_n^{\dagger}\gets0\) and \(P_n^{\dagger}\gets0\).
  \Else
    \State Compute \(A_n^{\dagger}=Q(\calU_n)/Q(\calR_n)\) and
    \(P_n^{\dagger}=Q_{1-\gamma}[U_n^p(X)-L_n^p(X)\mid X\in\calR_n]\).
  \EndIf
  \State Compute
  \(\Psi_n^{\dagger}=\max\{A_n^{\dagger}/\varepsilon,P_n^{\dagger}/d\}\).
  \If{\(\Psi_n^{\dagger}\le1\)}
    \State \Return \(T^{\dagger}=n\), \(\widehat p_n(\cdot)\), and all
    certified decisions.
  \Else
    \State Advance to the next scheduled look, if one remains.
  \EndIf
\EndWhile
\State \Return no certification by the final scheduled look; retain the terminal fit and diagnostics.
\end{algorithmic}
\end{algorithm}

The functional \(\Psi_n^{\mathrm{all}}\) measures ambiguity as a fraction of the
\emph{entire} target population.  Under extreme rarity this is too weak: as
information accrues, the overwhelming low-risk majority becomes certified below the
threshold, so \(A_n=Q(\calU_n)\) can be small and the risk interval can be narrow
where \(p_r(x)\approx0\) long before the few referral-relevant patients are
resolved.  The baseline may then fire with very few genuine events, certifying little
more than \emph{refer essentially no one}.

Our primary rule prevents this dilution by targeting the decision-relevant
population.  Split the certified set into its two decisions,
\[
\calC_n^{+}=\{x:L_n^\eta(x)>c\}\ (\text{certify refer}),\qquad
\calC_n^{-}=\{x:U_n^\eta(x)<c\}\ (\text{certify do-not-refer}),
\]
leaving \(\calU_n\) ambiguous, and let the \emph{referral-relevant} region be the
profiles not already certified below the threshold,
\[
\calR_n=\calU_n\cup\calC_n^{+}=\{x:U_n^\eta(x)\ge c\}.
\]
Define the decision-targeted ambiguity and width
\begin{equation}\label{eq:Adagger}
A_n^{\dagger}=\frac{Q(\calU_n)}{Q(\calR_n)},
\qquad
P_n^{\dagger}=Q_{1-\gamma}\!\left[\,U_n^p(X)-L_n^p(X)\ \middle|\ X\in\calR_n\right],
\end{equation}
when \(Q(\calR_n)>0\).  If \(Q(\calR_n)=0\), set
\(A_n^{\dagger}=P_n^{\dagger}=0\).  In this empty-set case every target profile has
already been certified below the referral threshold, so no referral-relevant width
remains to be controlled.  The proposed decision-targeted rule is
\begin{equation}\label{eq:Tdagger}
\Psi_n^{\dagger}
=\max\left\{\frac{A_n^{\dagger}}{\varepsilon},
             \frac{P_n^{\dagger}}{d}\right\},
\qquad
T^{\dagger}=\inf\left\{n\in\calN:\ n\ge n_{\min},\;
\Psi_n^{\dagger}\le1\right\}.
\end{equation}
Now \(A_n^{\dagger}\) is the ambiguous fraction \emph{among patients who might be
referred}, and \(P_n^{\dagger}\) controls precision only where a decision is at
stake.  When the threshold lies in the bulk of \(Q\) (so \(\calR_n=\calX_0\)) the
two rules coincide; the refinement matters precisely when the action concerns a
rare region.

\begin{proposition}[Relation between the baseline and decision-targeted rules]
\label{prop:rulecomparison}
At a monitored look with \(Q(\calR_n)>0\),
\(A_n^{\dagger}=A_n/Q(\calR_n)\ge A_n\).  If
\(\calR_n=\calX_0\) at every eligible look, then
\(\Psi_n^{\dagger}=\Psi_n^{\mathrm{all}}\) and
\(T^{\dagger}=T^{\mathrm{all}}\).  More generally, there is no unconditional
ordering of the stopping times because \(P_n^{\dagger}\) is a conditional quantile
and the empty-relevant-set convention can make it zero.  If
\(P_n^{\dagger}\ge P_n\) at every eligible look, however, then
\(\Psi_n^{\dagger}\ge\Psi_n^{\mathrm{all}}\) and
\(T^{\dagger}\ge T^{\mathrm{all}}\).  In particular, if the baseline stops at a
look for which \(A_n^{\dagger}>\varepsilon\) or \(P_n^{\dagger}>d\), that stop is
premature for the decision-targeted criterion.
\end{proposition}

\begin{proof}
The ambiguity inequality follows from \(0<Q(\calR_n)\le1\).  When
\(\calR_n=\calX_0\), the ambiguity masses and width quantiles coincide.  Under the
additional width ordering, both components of the maximum defining
\(\Psi_n^{\dagger}\) are at least their whole-population counterparts, which gives
the stopping-time ordering.  The final statement follows directly from the
definition of \(T^{\dagger}\).
\end{proof}

\begin{proposition}[Validity is inherited]
\label{prop:dtvalid}
Under the coverage hypothesis of Theorem~\ref{thm:cert}, every decision certified
at the prespecified-grid-selected time \(T^{\dagger}\) of
\eqref{eq:Tdagger}, whenever that time is finite, is model-conditionally
correct simultaneously with probability at least
\(1-\alpha\).
\end{proposition}

\noindent\emph{Argument (from Theorem~\ref{thm:cert}, no new proof needed).}
Theorem~\ref{thm:cert} guarantees that, on the simultaneous-coverage event, every
profile placed in \(\calC_n^{+}\cup\calC_n^{-}\) at every prespecified
monitored time
carries the correct decision.  The decision-targeted rule changes only the
\emph{stopping} criterion, not the per-profile certified set; since
\(T^{\dagger}\in\calN\) whenever it is finite, the one simultaneous coverage
event validates all decisions certified at that time. No decision is certified
by a rule that has not stopped.\hfill\(\square\)

\begin{remark}[Synthetic oversampling is not a source of information]
We do not devote a section to synthetic oversampling such as SMOTE; one
statistical observation suffices.  If synthetic data \(S\) are generated
conditionally from the observed data \(D\) by a mechanism \(q(s\mid d)\) that
does not depend on the unknown parameter, then
\(f_\theta(d,s)=f_\theta(d)\,q(s\mid d)\),
so the joint score equals the score from \(D\) alone.  Synthetic rows may change
an optimization objective, but they do not create new independent likelihood
information.  The proposed decision-targeted stopping rule therefore counts
only genuine accrued information.
\end{remark}

\begin{remark}[Why the refinement is not vacuous]
Under the whole-population baseline, \(A_n\le\varepsilon\) can hold while
\(\calR_n\) remains mostly ambiguous.  Normalising by \(Q(\calR_n)\)
removes this dilution.  Theorem~\ref{thm:rate} supplies an oracle
information benchmark for resolving an individual margin; it does not imply a
rate for \(T^\dagger\).  The population rule must also meet its conditional
width requirement, which can remain binding even after individual decisions
are resolved.
\end{remark}

\noindent Note that
for positive sequences \(a_r\) and \(b_r\), the notation
\(a_r\asymp b_r\) means that they have the same order in the asymptotic regime
under discussion: there exist constants \(0<c\le C<\infty\), independent of
\(r\), such that \(c\,b_r\le a_r\le C\,b_r\) for all sufficiently large
\(r\).  It does \emph{not} mean that \(a_r/b_r\to1\); ratio convergence to one
is denoted by \(a_r\sim b_r\).  When the index is \(n\), the same convention is
used with ``sufficiently large \(n\).''

\subsection{Theoretical properties}
\label{sec:theory}

We establish rare-event information scaling, simultaneous-grid
model-conditional validity, and an oracle first-passage result with explicit conditions for transfer to
pointwise certification.  Complete proofs are provided in
Appendix~\ref{app:proofs}, and the classical second-order stopping-time
development is provided in Appendix~\ref{app:secondorder}.

\subsubsection{Rare-event information scaling}

Let \(\widetilde X=(1,X^\top)^\top\).  In regime \(r\), the information after
\(n\) observations is
\(
\bI_{r,n}=\sum_{i=1}^n
p_r(X_{ri})\{1-p_r(X_{ri})\}\widetilde X_{ri}\widetilde X_{ri}^\top.
\)

\begin{assumption}[Rare-event moments]
\label{ass:mom}
The \(X_{ri}\) are i.i.d.\ within each regime,
\(\E[e^{X^\top\beta_0}\|\widetilde X\|^2]<\infty\), and
\(
\bJ=\E\left[e^{X^\top\beta_0}\widetilde X\widetilde X^\top\right]
\)
is positive definite on the effective model subspace.
\end{assumption}

\begin{proposition}[Effective rare-event information]
\label{prop:info}
Under Assumption~\ref{ass:mom} and a triangular-array law of large numbers,
\[
\frac{\bI_{r,n}}{n\rho_r}\overset{p}{\longrightarrow}\bJ
\qquad\text{as }r,n\to\infty\text{ with }n\rho_r\to\infty.
\]
\end{proposition}

The interpretation is
\(
\text{effective information size}\asymp n\rho_r.
\)

\subsubsection{Validity of certified classifications at the stopping time and pointwise certification complexity}
The theorem below is about decision certification, not convergence of a monitored
trajectory.

\begin{theorem}[Simultaneous-grid model-conditional decision certification]
\label{thm:cert}

Suppose the band \([L_n^\eta(x),U_n^\eta(x)]\) has simultaneous coverage
over every \((n,x)\in\calN\times\calX_0\) with probability at least
\(1-\alpha\). Let \(S\) be any data-dependent monitored look taking values in
\(\calN\cup\{\infty\}\), including \(T^{\mathrm{all}}\) or \(T^{\dagger}\),
where infinity denotes no stop. Then, with probability at least \(1-\alpha\),
whenever \(S<\infty\), every decision certified at \(S\) is
model-conditionally correct:
\[
x\in\calC_S
\quad\Longrightarrow\quad
\widehat\delta_S(x)=\delta_r(x)
\qquad\text{for every }x\in\calX_0.
\]
\end{theorem}


For a target profile \(x_r\) in regime \(r\), write
\(m_r=|\eta_r(x_r)-c_r|>0\), where \(c_r\) is the prespecified log-odds
threshold; a fixed profile or threshold is included when appropriate.  Define
\begin{equation}\label{eq:nstar}
n_{\star,r}=\frac{a_{\alpha,r}^2v_r}{\rho_r m_r^2},
\qquad
T^{\mathrm{or}}_{x,r}=\inf\left\{n\in\calN_r:n\ge n_{\min,r},\;
a_{\alpha,r}\sqrt{\widehat V_{r,n}(x_r)}<m_r\right\},
\end{equation}
where \(v_r>0\) and \(\inf\emptyset=\infty\).  The superscript ``or''
identifies an \emph{oracle information crossing}: its definition uses the unknown
true margin.  Actual pointwise certification instead occurs at
\[
\widehat T_{x,r}=\inf\left\{n\in\calN_r:n\ge n_{\min,r},\;
|\widehat\eta^{\mathrm{bc}}_{r,n}(x_r)-c_r|
>a_{\alpha,r}\sqrt{\widehat V_{r,n}(x_r)}\right\}.
\]
These times must not be identified without an additional control on predictor
estimation error.

\begin{theorem}[Oracle information crossing and pointwise certification]
\label{thm:rate}
Suppose \(n_{\star,r}\to\infty\),
\(n_{\min,r}/n_{\star,r}\to0\), and
\(n_{\min,r}\rho_r\to\infty\).  Assume the following conditions for every
fixed \(\delta\in(0,1)\).
\begin{enumerate}
\item The first scheduled look \(u_r(\delta)\) at or above
\((1+\delta)n_{\star,r}\) exists and
\(u_r(\delta)/n_{\star,r}\to1+\delta\).
\item The variance approximation is uniform over \emph{all eligible earlier
looks}, not only a neighbourhood of the crossing:
\begin{equation}\label{eq:rateuniform}
R_r(\delta):=
\sup_{\substack{n\in\calN_r\\n_{\min,r}\le n\le u_r(\delta)}}
\left|\frac{n\rho_r\widehat V_{r,n}(x_r)}{v_r}-1\right|
\xrightarrow{p}0.
\end{equation}
\end{enumerate}
Then \(T^{\mathrm{or}}_{x,r}/n_{\star,r}\xrightarrow{p}1\).
If, in addition, for every such \(\delta\),
\begin{equation}\label{eq:ratepredictor}
B_r(\delta):=
\sup_{\substack{n\in\calN_r\\n_{\min,r}\le n\le u_r(\delta)}}
\frac{|\widehat\eta^{\mathrm{bc}}_{r,n}(x_r)-\eta_r(x_r)|}{m_r}
\xrightarrow{p}0,
\end{equation}
then also \(\widehat T_{x,r}/n_{\star,r}\xrightarrow{p}1\).
\end{theorem}

\begin{remark}[Interpretation and scope of the rate]
Theorem~\ref{thm:rate} separates three requirements: enough effective
information before monitoring, exclusion of premature crossings, and a monitoring
grid that becomes dense on the \(n_{\star,r}\) scale.  The normal approximation
in Assumption~\ref{ass:estimator} alone does not imply
\eqref{eq:ratepredictor}: near \(n_{\star,r}\), the standard error divided by
\(m_r\) is approximately \(1/a_{\alpha,r}\), which need not vanish for a fixed
critical value.  Neither ratio limit is asserted for the fixed geometric grids
used in our numerical studies.  On a geometric grid with a fixed spacing ratio,
rounding alone can prevent ratio convergence to one.

The deterministic benchmark \(n_{\star,r}\) exposes the joint effects of
rarity, margin, and profile-specific variance.  When \(a_{\alpha,r}^2v_r\)
is bounded above and away from zero, it has order
\(1/(\rho_r m_r^2)\); this is an oracle information scale, not a proved rate
for the population stopping rule \(T^\dagger\), whose width criterion can also
be binding.  The assumptions require a compatible joint limit.  In particular,
with fixed \(x\), fixed \(\tau\), and bounded \(a_{\alpha,r}^2v_r\),
letting \(\rho_r\to0\) makes \(m_r\) grow and
\(n_{\star,r}\rho_r\to0\); that regime is excluded.  A vanishing-margin
regime with sufficiently large \(n_{\min,r}\rho_r\), for example, can satisfy
the oracle conditions.  Appendix~\ref{app:secondorder} instead uses a
fixed-regime precision limit.
\end{remark}

Under separate bounded-information assumptions,
Appendix~\ref{app:secondorder} develops an oracle normal
approximation, expected stopping-time correction, and \(\beta\)-protected design
quantile.  Those oracle calculations clarify the fluctuation scale but are not used
as validity evidence for the implementable finite-grid rule evaluated below.

\section{Simulation evaluation}
\label{sec:num}

\subsection{Simulation design and operating characteristics}
\label{sec:sim}

The revised simulation study evaluates the procedure actually defined in
Sections~\ref{sec:bandconstruct}--\ref{sec:decisiontargeted}.  We generate
\(X\sim N_5(0,\Sigma)\), with \(\Sigma_{jk}=0.3^{|j-k|}\), and use
\(\beta=(0.8,-0.6,0.5,-0.4,0.3)^\top\).  Unless otherwise stated,
\(Y\mid X\) follows the correctly specified logistic model, the intercept is
calibrated to the stated marginal prevalence \(\pi\), and the decision threshold is
\(\tau=0.10\).  We take \((\varepsilon,d,\gamma)=(0.10,0.10,0.10)\).
At every scheduled look the ridge fit uses an unpenalised intercept, \(\lambda=1\),
and the penalised sandwich covariance in Assumption~\ref{ass:estimator}.

For the rule-comparison study the target set contains \(M=60\) distinct profiles
and the schedule contains \(K=14\) looks.  Consequently every sequential band uses
\[
a_\alpha=\Phi^{-1}\!\left\{1-\frac{0.05}{2KM}\right\}=4.015,
\]
not the repeatedly read fixed-sample value 1.96.  The target distribution is
constructed so that a fraction \(q\in\{0.02,0.10\}\) lies above the referral
threshold at log-odds margin \(m\in\{0.30,0.75\}\), while the remainder is well
below it.  The same stream is used for all stopping rules within a replication.
There are 1,000 replications in every rule-comparison, monitoring, estimator, and
robustness cell.  We report restricted medians and interquartile ranges for censored
stopping times, censoring percentages, event-count means with Monte Carlo standard
errors, and Wilson 95\% intervals for coverage and error probabilities.  The
submission Data Files contain the frozen numerical summaries used for these tables
and figures together with a clean reference implementation of the documented
Study~1--4 certification mechanics and configurations.  The final historical
replication-level files and exact seeds were not retained, so byte-for-byte
regeneration of those particular Monte Carlo summaries is not claimed; where a
historical nuisance tuning value was not recorded, the reference implementation
exposes it explicitly rather than imputing provenance.

\subsection{Study 1: when does decision targeting matter?}

Table~\ref{tab:phase3-rules} gives the direct comparison that motivates the primary
rule.  The whole-population baseline can stop after the low-risk majority becomes
easy to classify: its median ranges from only 500 to 13,585 patients.  In contrast,
\(T^\dagger\) waits for uncertainty among patients who might actually be referred.
At \(\pi=0.02\), its restricted median is 42,268--65,874; at \(\pi=0.005\), it
reaches the 640,000-patient horizon in most scenarios.  In the two most difficult
near-threshold settings, 17.5\% and 29.5\% of replications remain uncertified at that
horizon.  Thus the low-risk majority can make \(T^{\mathrm{all}}\) smaller by one to
two orders of magnitude without resolving the clinically actionable part of the
target population.

\begin{table}[htbp]
\centering
\caption{Study 1: whole-population versus decision-targeted certification.  Stopping
times are restricted medians over the scheduled horizon; ``cens.'' is the
non-certification percentage for \(T^\dagger\).  Event counts are means (Monte Carlo
s.e.) among observed \(T^\dagger\) stops.  Coverage is simultaneous over all looks
and target profiles.}
\label{tab:phase3-rules}
\scriptsize
\begin{tabular}{rrrrrrrr}
\toprule
\(\pi\) & \(q\) & \(m\) & \(T^{\mathrm{all}}\) & \(T^\dagger\) & cens. & events at \(T^\dagger\) & coverage\\
\midrule
.005 & .02 & .30 & 2,606  & 640,000 & 17.5\% & 2,562 (23) & 99.4\%\\
.005 & .02 & .75 & 2,606  & 640,000 & 0\%    & 3,171 (7)  & 99.7\%\\
.005 & .10 & .30 & 7,835  & 640,000 & 29.5\% & 2,629 (25) & 99.1\%\\
.005 & .10 & .75 & 13,585 & 640,000 & 0\%    & 3,197 (3)  & 99.3\%\\
.020 & .02 & .30 & 779    & 42,268  & 0.1\%  & 1,144 (17) & 98.6\%\\
.020 & .02 & .75 & 500    & 65,874  & 0\%    & 1,267 (5)  & 98.7\%\\
.020 & .10 & .30 & 1,893  & 65,874  & 0.5\%  & 1,651 (22) & 99.4\%\\
.020 & .10 & .75 & 1,893  & 65,874  & 0\%    & 1,551 (11) & 99.5\%\\
\bottomrule
\end{tabular}
\end{table}

The Wilson lower limits for simultaneous coverage range from 97.7\% to 99.1\%,
consistent with the nominal guarantee and the conservatism expected from a
Bonferroni band over correlated profiles.  A coefficient-stability comparator and
a fixed-30-event comparator stop much earlier.  The fixed-event rule never satisfies
the decision-targeted criterion, while the coefficient-stability rule does so in
only 0.5\% of one scenario and never in the other seven (8000 paired replications
total).  Figure
\ref{fig:phase3-rules} displays the separation between the two population rules.

\begin{figure}[htbp]
\centering
\includegraphics[width=0.97\textwidth]{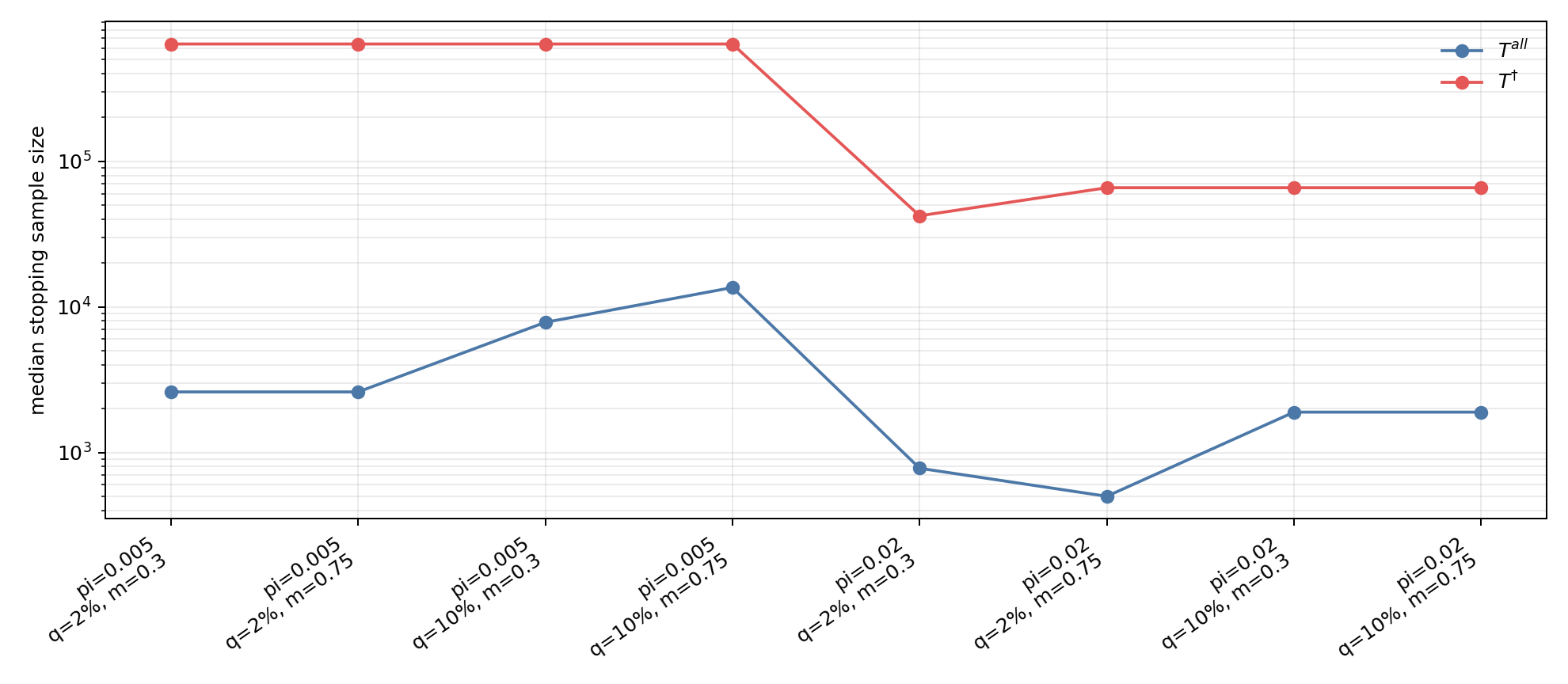}
\caption{Study 1: restricted median stopping sample size for the whole-population
baseline and the decision-targeted rule.  Here \(q\) is the high-risk fraction of
the target population and \(m\) is its log-odds distance above the referral
threshold.  The vertical scale is logarithmic.}
\label{fig:phase3-rules}
\end{figure}

\subsection{Study 2: estimable does not mean certifiable}

We next compare MLE, ridge, Firth, and FLIC on the same rare-event streams
(\(\pi=0.02\)).  All four algorithms return finite converged fits in this
low-dimensional experiment.  Their coefficient norms are already similar after
only about ten events, and nearly indistinguishable by \(n=8000\); nevertheless,
none of the 1,000 replications is decision-certified at \(n\le8000\).  At
\(n=64{,}000\), with about 1283 genuine events, the four estimators certify in
99.3--99.4\% of replications.  Thus the finding is not an artefact of one
penalisation method:
successful and stable-looking estimation occurs far earlier than simultaneous
decision certification.

The full estimator-by-look comparison, including Monte Carlo uncertainty, is
reported in Appendix~\ref{app:numerical}, Table~\ref{app:tabestimators}.

\subsection{Study 3: model-conditional coverage is not clinical truth}

The certification theorem is conditional on the working logistic model.  To make
that limitation empirical, Table~\ref{tab:phase3-misspec} compares correct
specification, an omitted centred quadratic term, a probit data-generating link,
and a shifted target covariate distribution.  ``Working coverage'' refers to the
pseudo-true logistic linear predictor, whereas false decisions are evaluated
against the actual data-generating risk.

Under correct specification, working coverage is 99.2\% and no false clinical
decision occurs.  Under the omitted quadratic, working coverage falls to 93.5\%
and every stopped replication contains at least one falsely certified target
profile, even though the mean false-decision fraction is only 1.67\% (one of 60
profiles).  The probit link is benign in this design, while covariate shift preserves
decision correctness but increases the median \(T^\dagger\) from 34,333 to 98,841.
These results support the precise claim made by the theory: the band certifies a
working-model target, and empirical calibration and specification checks remain
necessary for clinical interpretation.

\begin{table}[htbp]
\centering
\caption{Study 3: working-model coverage and actual decision error.  Parentheses
give 95\% Wilson intervals for probabilities.}
\label{tab:phase3-misspec}
\begin{tabular}{rrrrrr}
\toprule
setting & stop & median \(T^\dagger\) & working coverage & any false & mean false fraction\\
\midrule
correct   & 88.8\% & 34,333 & 99.2\% (98.4,99.6) & 0\% (0,0.43) & 0\%\\
quadratic & 95.8\% & 34,333 & 93.5\% (91.8,94.9) & 100\% (99.6,100) & 1.67\%\\
probit    & 99.5\% & 23,575 & 97.6\% (96.5,98.4) & 0\% (0,0.38) & 0\%\\
target shift & 100\% & 98,841 & 98.7\% (97.8,99.2) & 0\% (0,0.38) & 0\%\\
\bottomrule
\end{tabular}
\end{table}

\begin{figure}[htbp]
\centering
\includegraphics[width=0.82\textwidth]{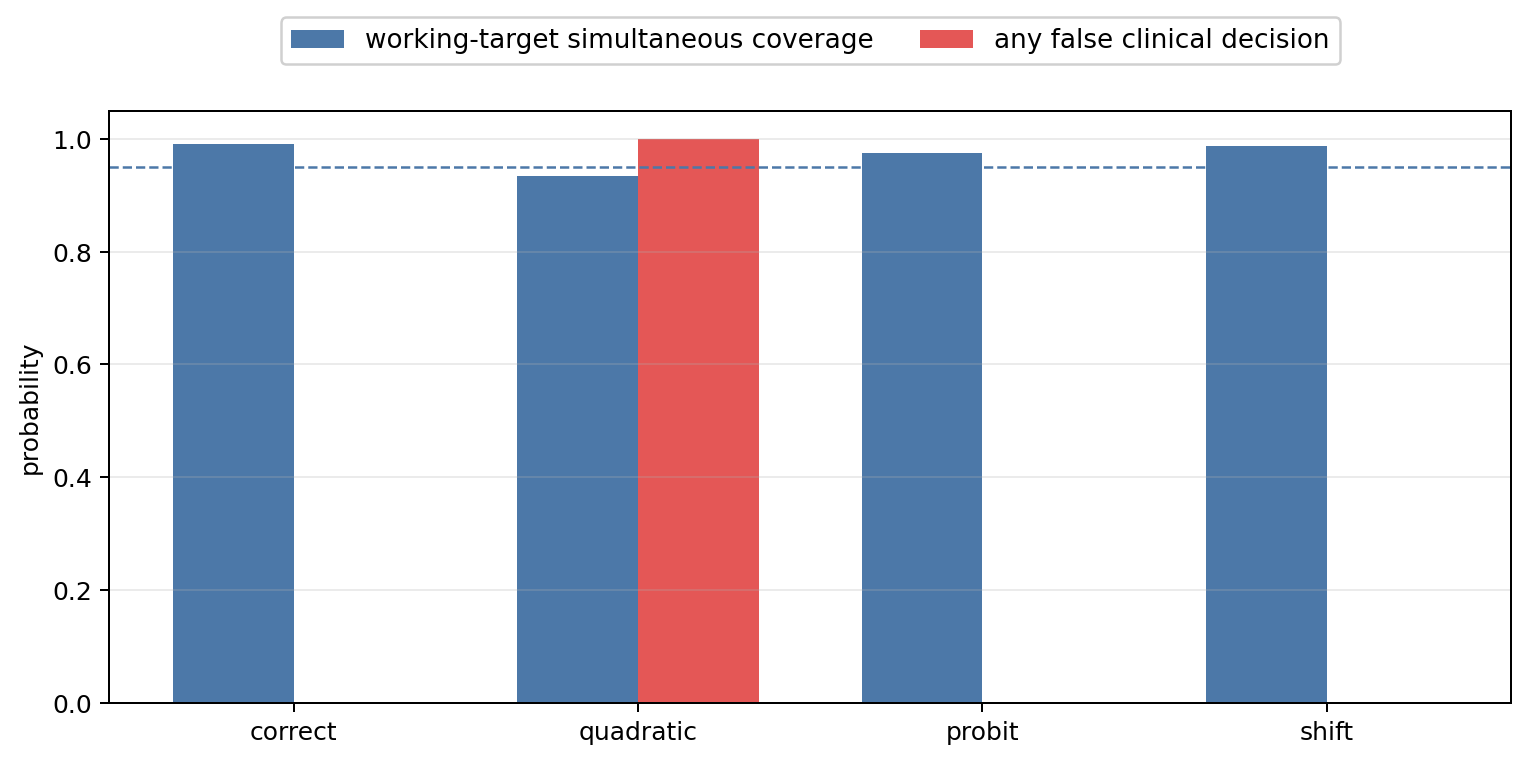}
\caption{Study 3: simultaneous coverage of the working-model target (blue) and
the probability of at least one false clinical decision at \(T^\dagger\) (red).
The dashed line marks 95\% working-target coverage.}
\label{fig:phase3-robustness}
\end{figure}

\subsection{Study 4: is repeated-monitoring adjustment necessary?}

We construct 40 distinct profiles exactly on the decision boundary and evaluate
the family-wise probability that at least one interval excludes that boundary over
the complete look-by-profile grid.  This is a simultaneous-coverage failure and
would produce an unjustified certification away from the boundary.  It is not
automatically a wrong binary classification because the convention in
\eqref{eq:decision} assigns equality to the positive class.  At one final look,
ordinary pointwise 1.96 intervals already have an 18.9\% family-wise
boundary-exclusion probability because 40 profiles are inspected.  Re-reading them
at 20 planned looks raises this probability to 76.2\%.  In contrast, the finite-grid
band remains between 0.4\% and 0.5\%; its upper Wilson limit never exceeds 1.2\%.
The conservatism reflects strong dependence among the nested looks and target
profiles, but the required simultaneous-coverage inequality is maintained.

\begin{table}[htbp]
\centering
\caption{Study 4: family-wise boundary-exclusion probability over all planned
looks and 40 target profiles (1,000 replications).  Intervals are 95\% Wilson
Monte Carlo intervals.}
\label{tab:phase3-monitoring}
\begin{tabular}{rrrr}
\toprule
looks & ordinary 1.96 & finite-grid critical value & finite-grid exclusion\\
\midrule
1  & 18.9\% (16.6,21.4) & 3.227 & 0.4\% (0.16,1.02)\\
4  & 53.9\% (50.8,57.0) & 3.605 & 0.4\% (0.16,1.02)\\
8  & 65.1\% (62.1,68.0) & 3.781 & 0.5\% (0.21,1.17)\\
12 & 70.6\% (67.7,73.3) & 3.881 & 0.4\% (0.16,1.02)\\
20 & 76.2\% (73.5,78.7) & 4.003 & 0.4\% (0.16,1.02)\\
\bottomrule
\end{tabular}
\end{table}

\begin{figure}[htbp]
\centering
\includegraphics[width=0.82\textwidth]{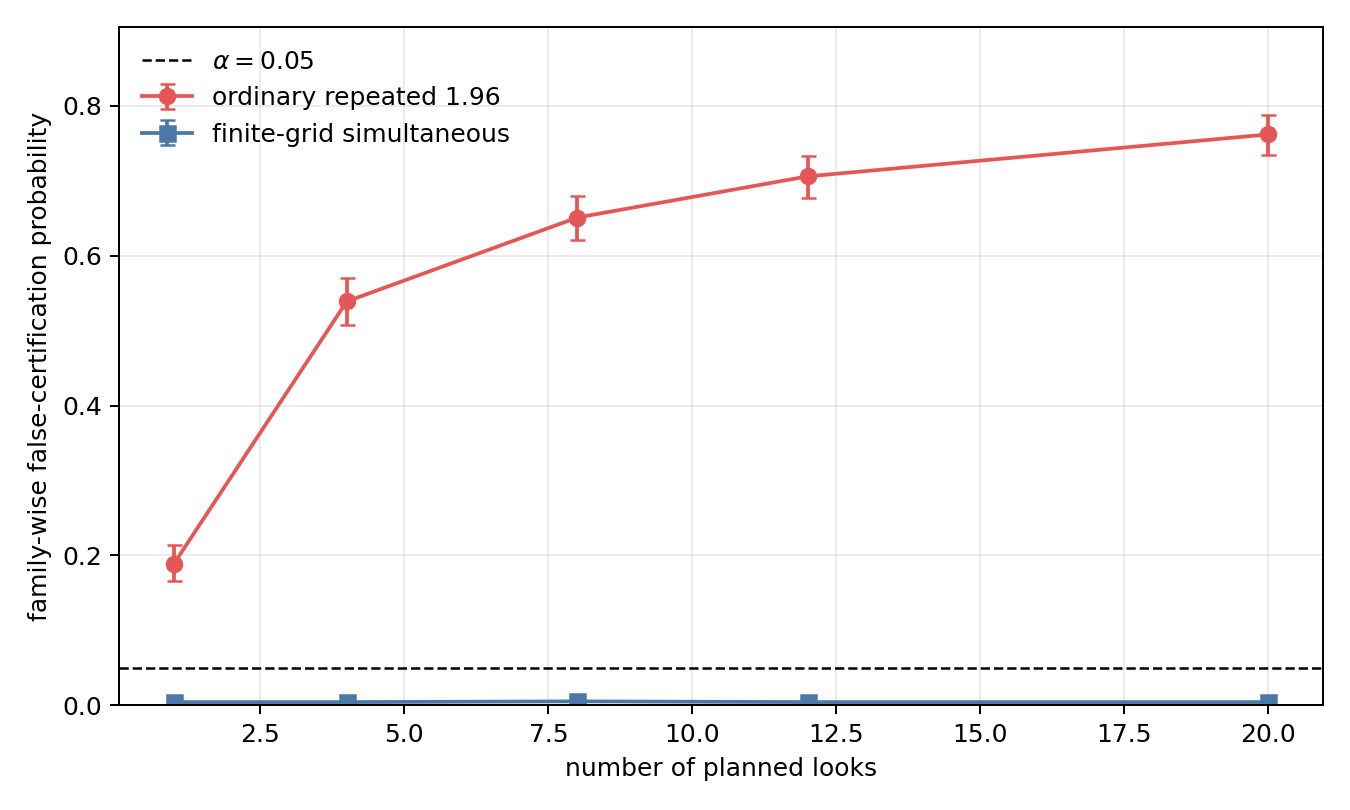}
\caption{Study 4: family-wise boundary-exclusion probability for repeatedly read
ordinary intervals and the prespecified finite-grid simultaneous band.  Error bars
are 95\% Wilson Monte Carlo intervals.}
\label{fig:phase3-monitoring}
\end{figure}

Earlier fixed-1.96 scaling plots are not used as evidence for sequential validity.
The documented oracle precision experiment in Appendix~\ref{app:numerical}
examines the stopping-time approximations only; it is not validity evidence for
the implementable finite-grid procedure.

\section{NCHS linked birth/infant-death application}
\label{sec:data}

The observed-data analysis asks whether the whole-population and
decision-targeted rules reach the same conclusion when infant-mortality records
accrue, and whether apparently stable external prediction performance is sufficient
for decision certification.  Development uses 2013 records, while a disjoint 2014
sample supplies temporal outcome validation.

\subsection{Data, target population, and monitoring design}
\label{sec:infant}

We used the official 2013 and 2014 US Cohort Linked Birth/Infant Death
public-use denominator files \cite{nchs2013guide,nchs2014guide}; the endpoint is linked
death within the first year of life.  We retained residence codes 1--3 and excluded code 4 (foreign residents).  A full fixed-width scan then gave
\(3{,}932{,}181\) resident births and \(23{,}142\) linked deaths in 2013
(\(0.5885\%\)), and \(3{,}988{,}076\) and \(23{,}256\), respectively, in
2014 (\(0.5831\%\)).  These counts refer to the complete resident cohorts, not to the sampled development and validation records.

The sequential development sample was a reproducible reservoir sample of
\(300{,}000\) 2013 resident births, containing \(1{,}660\) deaths.  The ridge
logistic model used twelve variables available at birth: birth weight,
gestational age, plurality, maternal-age group, infant sex, prenatal-care month
and prenatal-visit recode, maternal smoking, maternal education, and indicators for
Black, Hispanic, and other race/ethnicity.  The prenatal-visit variable is the
ordinal \texttt{PREVIS\_REC} category (01--11), not the raw visit count;
code 12 denotes missingness.  Unknown race/ethnicity makes all three indicators
missing before imputation, rather than assigning the reference category.
Official reporting flags and year-specific race/ethnicity codes were handled
by the same audited decoder in preparation and analysis.  Medians, centring,
and scaling were estimated once from the complete 300,000-record development
sample's covariates and then frozen; the ridge penalty was \(\lambda=1\), with
an unpenalised intercept and the penalised sandwich covariance specified above.

The target distribution \(Q\) was the empirical distribution of a fixed
\(4{,}000\)-profile sample from the 2014 cohort.  A disjoint 2014 sample of
\(100{,}000\) births with \(561\) deaths was reserved for temporal outcome
validation.  Target outcomes were not used in fitting or stopping and were read
only for post-stopping description.  Because the public-use files suppress
exact birth date, the twelve birth months were kept in chronological order and
records were randomised only within month.  We generated \(1{,}000\)
pseudo-streams.  Thus replication uncertainty describes unknown within-month
ordering, not independent population-sampling uncertainty.  This is a retrospective
experiment using completed infant-death follow-up and a fixed future-year
covariate target, not a simulation of when outcomes would become observable
in real time.  Prospective deployment would require an outcome-availability
schedule and an independently fixed or predictably updated preprocessing and
target-population specification.

We fitted the model at 18 prespecified looks from 500 through \(300{,}000\)
births.  To exclude the extremely sparse startup regime from certification,
the minimum eligible size was \(n_{\min}=14{,}783\), the first scheduled look
with at least 80 expected deaths under the audited 2013 event rate.  Earlier
fits were retained only as accrual diagnostics.  There were 54 nonconverged fits
at 500 births and 13 at 728 births; these were omitted from diagnostic summaries.
Every fit at 1,061 births and thereafter converged, including all eligible looks.
The two-sided normal band was
nevertheless Bonferroni-adjusted over all
\(18\times4{,}000=72{,}000\) fitted look--profile pairs
(\(\alpha=0.05\), critical value \(4.963\)), making it conservative for the
eligible subset.  The eligibility restriction is enforced for both rules in
all 27 settings.  All reported NCHS results were regenerated with this restriction
and the common decoder; no stopping-time row from an earlier implementation
is reused.

The width functional was the 90th percentile (\(\gamma=0.10\)).  Under the
primary specification \((\tau,d,\varepsilon)=(0.05,0.01,0.10)\), \(\tau=0.05\)
is an illustrative elevated-risk follow-up threshold rather than an established
clinical policy; \(d=0.01\) requires an absolute-risk interval width no greater
than one percentage point for at least 90\% of referral-relevant profiles; and
\(\varepsilon=0.10\) requires at least 90\% of that region to have a resolved
threshold decision.  These quantities should be elicited with clinicians and
fixed before monitoring in a prospective application.  Sensitivity analyses crossed
\(\tau\in\{0.01,0.025,0.05\}\),
\(d\in\{0.01,0.025,0.05\}\), and
\(\varepsilon\in\{0.05,0.10,0.20\}\).

\subsection{Stopping behavior and sensitivity}

The paired comparison exposes the imbalance problem directly
(Table~\ref{tab:phase4-main} and Figure~\ref{fig:phase4-stopping}).
Under the primary specification, \(T^{\mathrm{all}}\) certified in all
1,000 pseudo-streams at the 141,345-birth look, after a mean 808.2 genuine
deaths.  Its whole-population ambiguity and width met their tolerances
(\(A_{\mathrm{all}}=0.0134\), \(P_{\mathrm{all}}=0.00840\)), but among profiles
that might be referred the ambiguity was \(A_\dagger=0.568\) and the risk-band
width was \(P_\dagger=0.289\).  The whole-population stop therefore left
56.8\% of the referral-relevant region unresolved.

In contrast, \(T^\dagger\) did not certify in any primary pseudo-stream by the
prespecified horizon \(H=300{,}000\); this is right censoring, not a stopping
time of 300,000.  At \(H\), after all 1,660 sampled deaths, referral-relevant
ambiguity was still \(0.410\), width was \(0.213\), and
\(\Psi^\dagger(H)=21.27>1\).  No \(T^\dagger\) stop occurred in any of the 27
sensitivity cells; even the least restrictive cell had
\(\Psi^\dagger(H)=3.09\).  By comparison, \(T^{\mathrm{all}}\) stopped in 24
of the 27 cells.  

These results illustrate how acceptable whole-population precision can coexist
with substantial uncertainty among patients who might be referred.
By evaluating ambiguity and interval width within this group, the proposed rule
withheld certification while the referral-specific requirements remained unmet.
In this retrospective analysis, its value is to prevent the low-risk majority from
obscuring decision-relevant uncertainty;
the results do not establish improved predictive accuracy or clinical benefit.

{
\begin{table}[htbp]
\centering
\small
\caption{Paired NCHS stopping results under chronological birth-month accrual.}
\label{tab:phase4-main}
\resizebox{\textwidth}{!}{%
\begin{tabular}{lrrrrrrr}
\toprule
rule & stop rate & median $T_R$ & deaths & $A_{\rm all}$ & $A_\dagger$ & $P_\dagger$ & relevant fraction \\
\midrule
$T^{\mathrm{all}}$ & 1.000 & 141345 & 808.2 & 0.013 & 0.568 & 0.289 & 0.024 \\
$T^\dagger$ & 0.000 & 300000 & 1660.0 & 0.008 & 0.410 & 0.213 & 0.019 \\
\bottomrule
\end{tabular}}
\begin{minipage}{0.98\linewidth}\footnotesize Primary setting $(\tau,d,\varepsilon)=(0.05,0.01,0.10)$; $H=300{,}000$ and $T_R=\min(T,H)$. When a rule is censored, deaths and diagnostics are evaluated at $H$. Results use 1,000 pseudo-streams.\end{minipage}
\end{table}

\begin{table}[htbp]
\centering
\small
\caption{External 2014 predictive performance at stopping or, for a censored rule, at the horizon.}
\label{tab:phase4-validation}
\resizebox{\textwidth}{!}{%
\begin{tabular}{lrrrrrr}
\toprule
fit & calibration intercept & calibration slope & Brier & log loss & ROC-AUC & PR-AUC \\
\midrule
$T^{\mathrm{all}}$ & -0.008 & 0.965 & 0.0043 & 0.0232 & 0.852 & 0.372 \\
$T^\dagger$ horizon (censored) & 0.009 & 0.965 & 0.0043 & 0.0231 & 0.854 & 0.373 \\
full development sample & 0.009 & 0.965 & 0.0043 & 0.0231 & 0.854 & 0.373 \\
\bottomrule
\end{tabular}}
\begin{minipage}{0.98\linewidth}\footnotesize Calibration-in-the-large is the intercept from a slope-one offset model; the calibration slope is from a separate two-parameter recalibration.\end{minipage}
\end{table}
}

\begin{figure}[htbp]
\centering
\includegraphics[width=0.98\textwidth]{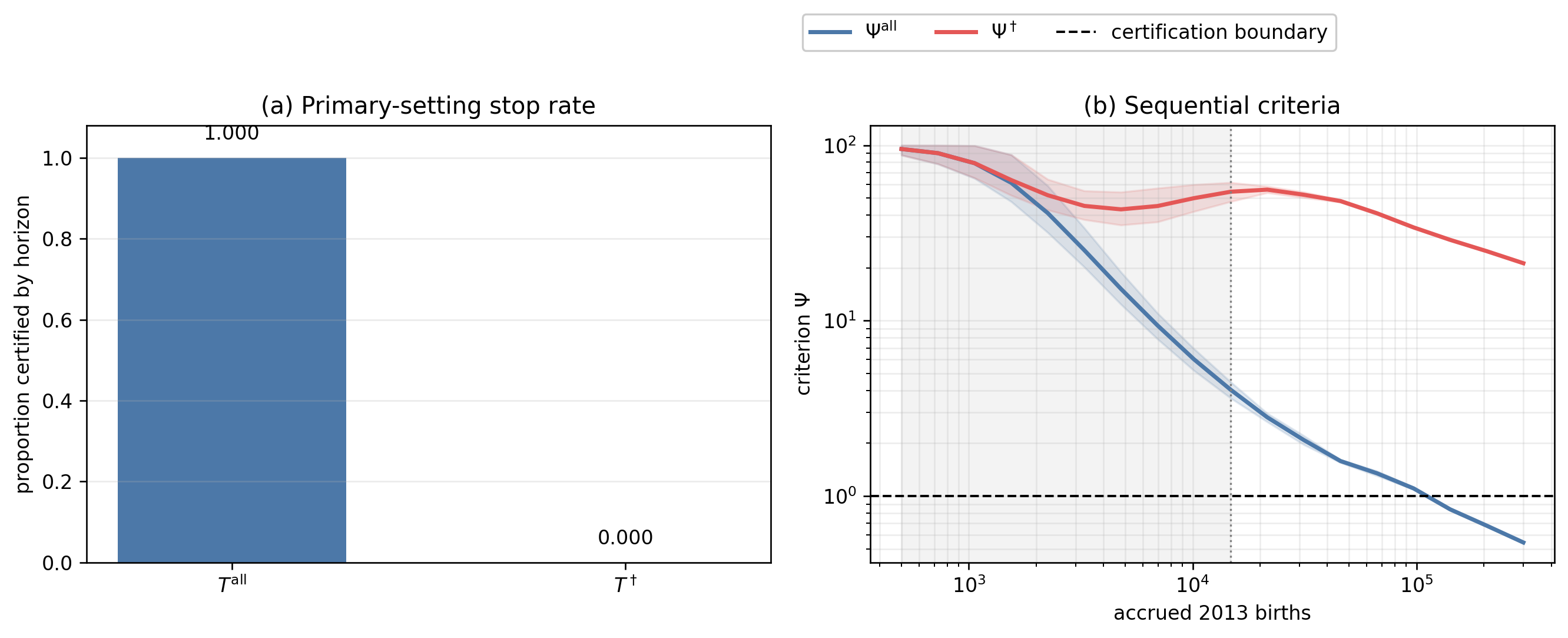}
\caption{Paired NCHS stopping behavior for the primary specification.
Panel (a) reports the fraction certified by \(H=300{,}000\).  Panel (b) shows
the median sequential criteria with 2.5th--97.5th percentile bands over 1,000
within-month pseudo-streams; certification requires \(\Psi\le1\) at an eligible look. The grey startup region precedes \(n_{\min}=14{,}783\).}
\label{fig:phase4-stopping}
\end{figure}

The complete 27-cell threshold--tolerance table and terminal-criterion
heatmap are reported in Appendix~\ref{app:nchs}, Table~\ref{app:tabgrid} and
Figure~\ref{app:figsensitivity}.

\subsection{Temporal validation and limitations}

Temporal validation gives a complementary result
(Table~\ref{tab:phase4-validation} and
Figure~\ref{fig:phase4-validation}; threshold-specific results are in
Appendix~\ref{app:nchs}, Table~\ref{app:tabthreshold}).  At
\(T^{\mathrm{all}}\), the 2014
calibration intercept was \(-0.008\), calibration slope \(0.965\), Brier score
\(0.00427\), log loss \(0.02315\), ROC-AUC \(0.852\), and PR-AUC \(0.372\).
At \(\tau=0.05\), sensitivity was \(0.501\), specificity \(0.989\), PPV
\(0.197\), NPV \(0.997\), and net benefit \(0.00221\), compared with zero for
treat-none and \(-0.0467\) for treat-all.  The horizon/full-development-sample fit produced
very similar point-prediction metrics (ROC-AUC \(0.854\), PR-AUC \(0.373\)).
The juxtaposition is substantive: apparently stable external point-prediction
performance does not imply that simultaneous uncertainty is narrow enough to
certify referral-region decisions.

\begin{figure}[htbp]
\centering
\includegraphics[width=0.98\textwidth]{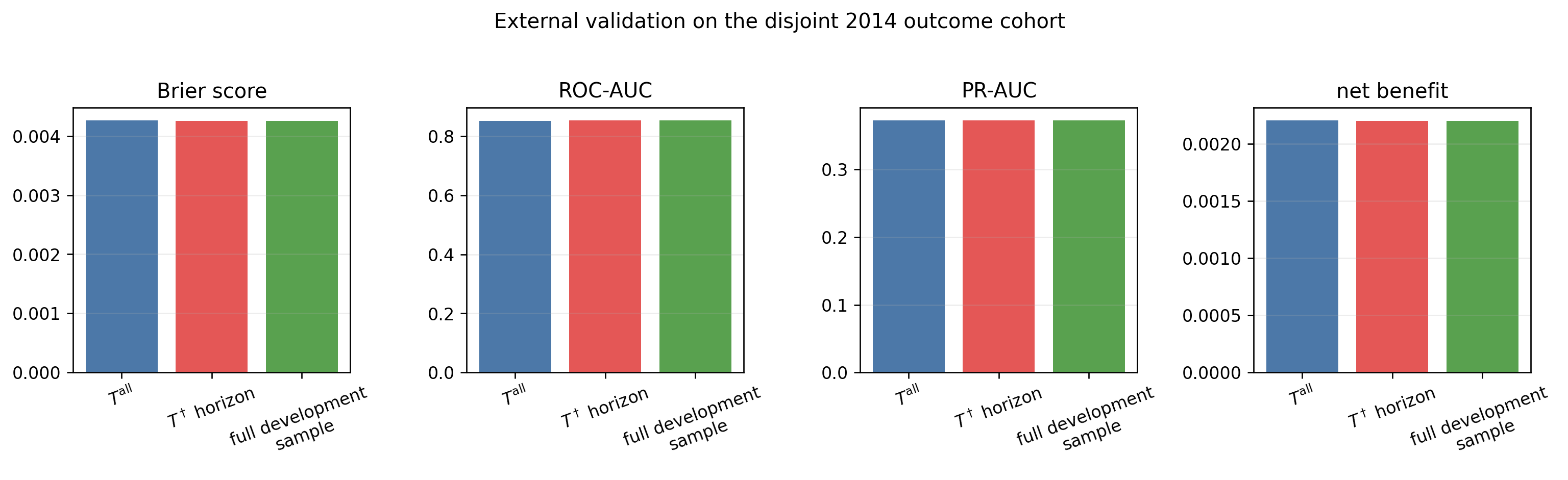}
\caption{Predictive performance on the disjoint 2014 outcome-validation
sample.  \(T^\dagger\) is shown at the study horizon because it never stopped;
the full-development-sample fit is a model-reproducibility comparator, not independent
evidence of clinical correctness.}
\label{fig:phase4-validation}
\end{figure}

The application has five important limitations.  First, pseudo-streams preserve
birth month but cannot recover suppressed within-month chronology or actual
outcome-availability times; completed outcomes and frozen full-sample covariate
preprocessing make the analysis retrospective.  Second,
the stopping horizon is the sampled \(300{,}000\) births, so the study provides
only the lower bound \(T^\dagger>300{,}000\) under the examined design.  Third,
linked public-use data may omit deaths that could not be linked to a birth
record.  Fourth, the threshold and predictor model are methodological examples,
not a validated clinical referral policy.  Fifth, the minimum eligible size
removes the most severely data-sparse looks but does not turn the asymptotic
normal band into an exact finite-sample procedure; estimator regularity and
calibration must still be assessed in each application.  These cautions prevent the
non-stopping result from being overstated while preserving its central message:
class-imbalance dilution can make whole-population certification materially
misleading.

A MODY-calibrated plasmode stress test, retained explicitly as simulated
rather than observed-data evidence, is reported in Appendix~\ref{app:mody}.

\section{Discussion}
\label{sec:discussion}

For a rare medical outcome the meaningful sample size is not the raw number of
patients but the number of informative rare-event contributions, which under the
rare-event logistic model appears through the effective scale
\(n\rho_r\).  The
procedure proposed here uses that information to answer a practical clinical
question---whether the predicted risks, and the decisions taken from them at a
prespecified threshold, are already precise enough to act on---without creating
synthetic events and without requiring precise estimation of every regression
coefficient.  It moves fixed-precision stopping from the coefficient vector to a
decision-targeted certification over a clinical target population.
The whole-population rule is a baseline; the primary rule \(T^{\dagger}\)
controls ambiguity and width among patients who might be referred.  A single
simultaneous predictive band over prespecified monitoring looks supplies both
estimation precision and model-conditional decision certification at the selected
look.  The present guarantee is finite-grid and is not an unrestricted
anytime-validity claim.  The absolute stopping times also depend on the chosen
simultaneous-band construction.  Bonferroni deliberately ignores dependence among
nested looks and target profiles; a calibrated max-statistic or multiplier/bootstrap
construction could reduce this conservatism while preserving the same prespecified-grid
principle.  Accordingly, the NCHS finding that \(T^\dagger>300{,}000\) should be read
as a conservative certification benchmark under the present band, not as an intrinsic
minimum cohort size.
In 1,000-replication experiments per cell, the whole-population baseline stops
one to two orders of magnitude before \(T^{\dagger}\) in rare actionable regions,
while coefficient stability or 30 accrued events almost never meet the
decision-targeted criterion.  The misspecification study further shows why
working-model coverage must not be interpreted as calibration to the true clinical
risk.
The temporally validated NCHS analysis makes the same distinction in
observed data: \(T^{\mathrm{all}}\) stopped at the 141,345-birth look while
56.8\% of referral-relevant profiles remained ambiguous, whereas
\(T^\dagger\) was right-censored at 300,000 births in every tolerance-grid
cell.  The 2014 point-prediction metrics were already close to their
full-development-sample values, showing that stable discrimination and
calibration summaries do not by themselves establish decision
certifiability.
The oracle information benchmark depends on rarity and the decision margin
through \(n_{\star,r}=a_{\alpha,r}^2v_r/(\rho_r m_r^2)\).  Theorem~\ref{thm:rate}
states the additional conditions needed for first-passage ratio limits; these are
not asserted for the fixed-grid population rule.  In practice, the simulations and the infant-mortality
application illustrate that a regularized model can look estimable well before its
absolute-risk decisions are certifiable.  Regularization, separation diagnostics,
and stability screening support this account rather than becoming parallel stories.

We have deliberately kept the scope narrow.  Ridge is used for stability rather
than as an optimal estimator; its inferential use requires the stated
small-penalty condition or an explicit penalty-bias correction, and all certification
claims remain conditional on the working risk model.  Unrestricted \(p\gg n\)
estimation, the boundary
degeneracy of near-zero or near-one probabilities and its reverse-martingale
characterization, and cost-based optimal stopping are separate problems treated
elsewhere or left for later.  Natural extensions include more general joint
rare-event, vanishing-margin, and growing-target-set asymptotics; high-dimensional certification
under sparsity, adaptive or surrogate-guided acquisition of genuine events, and
application to population-scale genomic and electronic-health-record cohorts once
the corresponding controlled-access data can be used.  A further direction is to
replace absolute-risk thresholds by treatment-specific utilities or individualized
causal contrasts; that extension would require causal identification, longitudinal
updating, and adaptive-regime validation, so the present paper should be read as a
prediction-to-decision certification framework rather than a full adaptive-treatment
theory.

Appendices~\ref{app:secondorder}--\ref{app:mody} provide the second-order
oracle precision development and \(\beta\)-protected design calculation, complete
proofs, detailed estimator and NCHS sensitivity tables, and the MODY-calibrated
plasmode.  Placing these supporting results after the Discussion keeps the main
text centered on the clinical decision problem, the implementable rule, and its
observed-data evidence while retaining a single self-contained manuscript.

\appendix
\renewcommand{\thesection}{\Alph{section}}
\renewcommand{\thesubsection}{\Alph{section}.\arabic{subsection}}
\makeatletter
\@addtoreset{equation}{section}
\@addtoreset{table}{section}
\@addtoreset{figure}{section}
\@addtoreset{theorem}{section}
\@addtoreset{proposition}{section}
\@addtoreset{lemma}{section}
\@addtoreset{corollary}{section}
\@addtoreset{remark}{section}
\@addtoreset{definition}{section}
\@addtoreset{assumption}{section}
\renewcommand{\theequation}{\Alph{section}.\arabic{equation}}
\renewcommand{\thetable}{\Alph{section}\arabic{table}}
\renewcommand{\thefigure}{\Alph{section}\arabic{figure}}
\renewcommand{\thetheorem}{\Alph{section}.\arabic{theorem}}
\renewcommand{\theproposition}{\Alph{section}.\arabic{proposition}}
\renewcommand{\thelemma}{\Alph{section}.\arabic{lemma}}
\renewcommand{\thecorollary}{\Alph{section}.\arabic{corollary}}
\renewcommand{\theremark}{\Alph{section}.\arabic{remark}}
\renewcommand{\thedefinition}{\Alph{section}.\arabic{definition}}
\renewcommand{\theassumption}{\Alph{section}.\arabic{assumption}}
\makeatother

\section{Second-order oracle precision theory}
\label{app:secondorder}

This appendix fixes the rare-event regime and studies a vanishing precision
requirement. It concerns \emph{oracle} Fisher-information weights at every integer
sample size, not the fitted sandwich variance or the fixed-grid population rule
\(T^\dagger\). These distinctions are needed before applying classical
stopping-time ideas \cite{woodroofe1977,lai1979,woodroofe1982,aras1993}.
Write \(z=\widetilde x\ne0\),
\(W_i=w(X_i)\widetilde X_i\widetilde X_i^\top\), and
\(J=\E W_1\), where \(w(x)=p(x)\{1-p(x)\}\) uses the true parameter of the
fixed working model. Define
\[
u=J^{-1}z,\qquad q=z^\top J^{-1}z=\sigma^2(x)>0,
\qquad I_n=\sum_{i=1}^n W_i.
\]

\begin{assumption}[Bounded oracle information]
\label{app:assbounded}
The matrices \(W_i\) are i.i.d., symmetric and positive semidefinite in a fixed
finite dimension, \(\|W_i\|_{\rm op}\le L<\infty\) almost surely, and
\(J\) is positive definite. The profile \(z\) and critical value
\(a_\alpha>0\) are fixed.
\end{assumption}

For example, bounded covariates in a fixed logistic model satisfy the boundedness
condition. Set \(V_n=z^\top I_n^{-1}z\) when \(I_n\) is positive definite,
and \(V_n=\infty\) otherwise. Singular information therefore cannot certify.
For \(h>0\), define
\begin{equation}\label{app:oraclecross}
K_n=\frac{q}{V_n},\qquad
N_h=\inf\{n\ge1:a_\alpha\sqrt{V_n}\le h\}
    =\inf\{n\ge1:K_n\ge n_0\},\qquad
n_0=\frac{a_\alpha^2q}{h^2},
\end{equation}
with \(q/\infty=0\). Here \(h\) is a precision half-width, not the referral
threshold \(c\) in the main text. The sequence \(K_n\) is nondecreasing because
information matrices increase in the positive-semidefinite order.

The constants governing the expansions are
\begin{equation}\label{app:constants}
\zeta^2=\Var(u^\top W_1u),\qquad
\delta=u^\top\E[(W_1-J)J^{-1}(W_1-J)]u,\qquad
\kappa^2=\frac{\zeta^2}{q^2}.
\end{equation}
All are finite under Assumption~\ref{app:assbounded}.

\begin{theorem}[Oracle first-passage limit]
\label{app:thmclt}
Under Assumption~\ref{app:assbounded}, \(N_h/n_0\to1\) almost surely as
\(h\downarrow0\). If \(\zeta^2>0\), then
\[
\frac{N_h-n_0}{\sqrt{n_0}}\xrightarrow{d}N(0,\kappa^2).
\]
\end{theorem}

\begin{proposition}[Expected oracle stopping time]
\label{app:propregret}
Under Assumption~\ref{app:assbounded}, let
\(O_h=K_{N_h}-n_0\ge0\) be the overshoot in effective-information units.
Then \(\E O_h=O(1)\), and
\begin{equation}\label{app:meanexpansion}
\E N_h=n_0+\frac{\delta}{q}-\frac{\zeta^2}{q^2}
              +\E O_h+o(1).
\end{equation}
If, additionally, \(\E O_h\to\nu\), this becomes
\(\E N_h=n_0+\delta/q-\zeta^2/q^2+\nu+o(1)\).
\end{proposition}

\begin{remark}[Stopping changes the second-order correction]
\label{app:remselection}
The term \(-\zeta^2/q^2\) is essential: substituting a fixed-sample expansion
at a random stopping time would omit it. In the scalar-information case,
\(\delta/q=\zeta^2/q^2\), leaving only the overshoot correction, as for an
ordinary positive-increment renewal process. A constant limit \(\nu\) is an
additional assumption; bounded information alone does not remove possible
arithmetic oscillations. Formula~\eqref{app:meanexpansion} retains the actual
mean overshoot and does not require such a limit. Because the weights here are
oracle weights, this result does not establish a regret correction for fitted
ridge or Firth estimators.
\end{remark}

\begin{definition}[\(\beta\)-protected oracle precision]
\label{app:defbeta}
For fixed \(\beta\in(0,1)\), an integer sample size \(n\) is
\((h,\beta)\)-protected at \(x\) if
\(\Prb\{a_\alpha\sqrt{V_n}\le h\}\ge1-\beta\).
\end{definition}

\begin{proposition}[Protected sample size and its normal approximation]
\label{app:propbeta}
Under Assumption~\ref{app:assbounded}, the smallest protected sample size is
exactly \(n_\beta(h)=Q_{1-\beta}(N_h)\), where the quantile is the smallest
integer attaining probability at least \(1-\beta\). If \(\zeta^2>0\),
\begin{equation}\label{app:protected}
n_\beta(h)=n_0+z_{1-\beta}\kappa\sqrt{n_0}+o(\sqrt{n_0}),
\qquad z_{1-\beta}=\Phi^{-1}(1-\beta).
\end{equation}
The rounded approximation
\(\widetilde n_\beta=\lceil n_0+z_{1-\beta}\kappa\sqrt{n_0}\rceil\)
satisfies
\(\Prb\{a_\alpha\sqrt{V_{\widetilde n_\beta}}\le h\}\to1-\beta\);
it is not an exact finite-sample protection guarantee.
\end{proposition}

The coverage level \(1-\alpha\) and planning level \(1-\beta\) have different
roles: the former concerns a predictive band, whereas the latter budgets
observations for the oracle precision target \cite{chenwangchang2011}.
No second-order claim for \(T^\dagger\) follows from this appendix.

\section{Proofs}
\label{app:proofs}

\subsection{Effective rare-event information}

For fixed $x$,
\[
\frac{p_r(x)\{1-p_r(x)\}}{\rho_r}
=\frac{e^{x^\top\beta_0}}
       {\{1+\rho_r e^{x^\top\beta_0}\}^2}
\longrightarrow e^{x^\top\beta_0}.
\]
Moreover,
\[
0\le \frac{p_r(x)\{1-p_r(x)\}}{\rho_r}\|\widetilde x\|^2
\le e^{x^\top\beta_0}\|\widetilde x\|^2.
\]
The moment assumption in the main paper makes the upper bound integrable.
Dominated convergence yields convergence of the expected normalized information,
and the triangular-array law of large numbers gives
\(\bI_{r,n}/(n\rho_r)\xrightarrow{p}\bJ\).

\subsection{Simultaneous-grid decision certification}

Work on the simultaneous coverage event, whose probability is at least
$1-\alpha$.  For any selected $S\in\calN$, if
$U_S^\eta(x)<c$, then
\(\eta_r(x)\le U_S^\eta(x)<c\), so the working-model decision is zero.  If
$L_S^\eta(x)>c$, then
\(\eta_r(x)\ge L_S^\eta(x)>c\), so the decision is one.  Because the same
coverage event holds over every prespecified time--profile pair, the implication
remains true when a finite $S$ is selected from \(\calN\) using the accrued data.
When \(S=\infty\), the statement makes no certification assertion. The
argument neither classifies uncertified profiles nor extends beyond the
prespecified schedule.

\subsection{Oracle information crossing and pointwise certification}

Fix \(\delta\in(0,1)\) and abbreviate \(n_\star=n_{\star,r}\),
\(u=u_r(\delta)\), \(R=R_r(\delta)\), and \(m=m_r\).
For every eligible scheduled \(n\le(1-\delta)n_\star\),
\[
\frac{a_{\alpha,r}^2\widehat V_{r,n}(x_r)}{m^2}
=\frac{n_\star}{n}
  \frac{n\rho_r\widehat V_{r,n}(x_r)}{v_r}
\ge\frac{1-R}{1-\delta}.
\]
Because \(R\to_p0\), the right-hand side exceeds one with probability
 tending to one.  This is a simultaneous bound over \emph{every earlier
eligible look}, so it rules out any earlier oracle crossing without assuming
that the estimated variance is monotone.  At \(u\),
\[
\frac{a_{\alpha,r}^2\widehat V_{r,u}(x_r)}{m^2}
=\frac{n_\star}{u}\{1+o_p(1)\}
\xrightarrow{p}\frac{1}{1+\delta}<1.
\]
Thus, with probability tending to one,
\((1-\delta)n_\star<T^{\mathrm{or}}_{x,r}\le u\).
The grid condition and arbitrarily small \(\delta\) give
\(T^{\mathrm{or}}_{x,r}/n_\star\to_p1\); in particular, non-crossing by the
horizon has probability tending to zero under these assumptions.

For the additional conclusion, let \(B=B_r(\delta)\).  The reverse triangle
inequality gives, uniformly on the same eligible looks,
\[
1-B\le
\frac{|\widehat\eta^{\mathrm{bc}}_{r,n}(x_r)-c_r|}{m}
\le1+B.
\]
Since \(B\to_p0\), choose fixed small bounds on \(R\) and \(B\) so that
\(\sqrt{(1-R)/(1-\delta)}>1+B\).  The earlier variance bound then excludes
\emph{every} actual certification before \((1-\delta)n_\star\).
At \(u\), the normalized half-width converges to
\((1+\delta)^{-1/2}<1\), whereas the normalized fitted margin converges to
one.  Actual certification therefore occurs no later than \(u\) with
probability tending to one.  The same bracket proves
\(\widehat T_{x,r}/n_\star\to_p1\).

\subsection{Oracle precision results}
\label{app:prooforacle}

Throughout this subsection, constants are independent of \(h\) and \(n\),
and matrix norms can be taken to be Frobenius norms because the dimension is
fixed. Write
\[
A_n=\sum_{i=1}^n(W_i-J),\qquad
S_n=u^\top A_nu,\qquad Q_n=u^\top A_nJ^{-1}A_nu.
\]

\paragraph{Information expansion and stopping-time bounds.}
On \(\|A_n/n\|_{\rm op}\le\lambda_{\min}(J)/2\), Taylor expansion of the
smooth map \(B\mapsto q/[z^\top(J+B)^{-1}z]\) gives
\begin{equation}\label{app:Kexpansion}
K_n=n+\frac{S_n}{q}
       +\frac{1}{n}\left(\frac{S_n^2}{q^2}-\frac{Q_n}{q}\right)+R_n,
\qquad |R_n|\le C\frac{\|A_n\|^3}{n^2}.
\end{equation}
Indeed, the expansion of the denominator is
\(q-S_n/n+Q_n/n^2+O(\|A_n/n\|^3)\), and expanding its reciprocal gives
both quadratic terms in \eqref{app:Kexpansion}. Define \(R_n\) by the equality
on all outcomes, including singular \(I_n\). Since \(I_n\preceq nL I\),
\(0\le K_n\le qLn/\|z\|^2\); boundedness of the increments also gives
\(|S_n|\le Cn\), \(Q_n\le Cn^2\), and hence \(|R_n|\le Cn\) globally.

The same bound on \(K_n\) implies the deterministic lower bound
\(N_h\ge b n_0\), where \(b=\|z\|^2/(qL)>0\). Bounded scalar
concentration applied to the finitely many entries of \(A_n\) gives
\begin{equation}\label{app:concentration}
\Prb\{\|A_n/n\|_{\rm op}>\lambda_{\min}(J)/2\}\le C_1e^{-C_2n}.
\end{equation}
On the complementary event,
\(I_n\succeq nJ/2\), so
\(K_n\ge qn\lambda_{\min}(J)/(2\|z\|^2)\).
Consequently, for a sufficiently large fixed \(C_0\),
\begin{equation}\label{app:Ntale}
\Prb(N_h>n)\le C_1e^{-C_2n}\quad(n\ge C_0n_0),
\qquad \E N_h^k=O(n_0^k)\quad(k\ge1\text{ fixed}).
\end{equation}
In particular, \(N_h\) is finite almost surely and integrable.

\paragraph{Proof of Theorem~\ref{app:thmclt}.}
The strong law gives \(K_n/n\to1\) almost surely. Monotonicity of \(K_n\)
then brackets its first crossing between
\(\lfloor(1-\epsilon)n_0\rfloor\) and
\(\lceil(1+\epsilon)n_0\rceil\), eventually for every fixed
\(\epsilon>0\), proving \(N_h/n_0\to1\).
At deterministic \(n\to\infty\), \(A_n=O_p(\sqrt n)\), so
\eqref{app:Kexpansion} gives \(K_n=n+S_n/q+O_p(1)\).
For \(n_t=\lfloor n_0+t\sqrt{n_0}\rfloor\), the ordinary i.i.d. central
limit theorem therefore yields
\[
\frac{K_{n_t}-n_0}{\sqrt{n_0}}\xrightarrow{d}N(t,\kappa^2).
\]
Because \(\{N_h\le n_t\}=\{K_{n_t}\ge n_0\}\) exactly, and the limiting
normal law is continuous when \(\kappa>0\),
\(\Prb\{(N_h-n_0)/\sqrt{n_0}\le t\}\to\Phi(t/\kappa)\).
This proves the stated limit without a random-time substitution.

\paragraph{Stopped quadratic moments and remainder.}
Put \(N=N_h\). The ratio limit and the maximal inequality for centered i.i.d.
partial sums imply
\[
\frac{A_N-A_{\lfloor n_0\rfloor}}{\sqrt{n_0}}\xrightarrow{p}0.
\]
For completeness, restrict first to
\(|N-\lfloor n_0\rfloor|\le\epsilon n_0\).
The maximal inequality on the forward and backward windows bounds the
probability of a fluctuation larger than \(a\sqrt{n_0}\) by
\(C\epsilon/a^2\); the probability of leaving this window tends to zero.
Then let \(\epsilon\downarrow0\). Thus \(A_N/\sqrt N\) has the same Gaussian
matrix limit as \(A_n/\sqrt n\).

To justify moment convergence, let \(m=\lceil C_0n_0\rceil\), increasing
\(C_0\) if necessary. Doob's fourth-moment inequality and bounded centered
increments give
\(\E\max_{k\le m}\|A_k\|^4\le C\E\|A_m\|^4=O(n_0^2)\).
On \(N>m\), use \(\|A_N\|\le CN\) and the exponential tail in
\eqref{app:Ntale}. It follows that
\(\E\|A_N\|^4=O(n_0^2)\).
Together with \(N\ge bn_0\), this makes \(Q_N/N\) and \(S_N^2/N\)
uniformly integrable. Their expectations consequently converge to the
corresponding Gaussian quadratic moments:
\begin{equation}\label{app:stoppedmoments}
\E\frac{Q_N}{N}\longrightarrow\delta,
\qquad \E\frac{S_N^2}{N}\longrightarrow\zeta^2.
\end{equation}
These limits also follow by evaluating the covariance of the centered matrix
increment \(W_1-J\); at deterministic \(n\),
\(\E Q_n=n\delta\) and \(\E S_n^2=n\zeta^2\).

Let \(E_h\) be the event that the matrix bound used in
\eqref{app:Kexpansion} holds for every \(n\ge\lfloor bn_0/2\rfloor\).
Summing \eqref{app:concentration} gives
\(\Prb(E_h^c)\le Ce^{-c_1n_0}\).
On \(E_h\), the cubic remainder and fourth-moment bound imply
\[
\E[|R_N|\mathbf1_{E_h}]
\le Cn_0^{-2}\E\|A_N\|^3=O(n_0^{-1/2}).
\]
On \(E_h^c\), the global bound and Cauchy--Schwarz give
\(\E[|R_N|\mathbf1_{E_h^c}]
\le C(\E N^2)^{1/2}\Prb(E_h^c)^{1/2}=o(1)\).
Hence \(\E|R_N|\to0\).

\paragraph{Proof of Proposition~\ref{app:propregret}.}
The centered increments of \(S_n\) are bounded, and
\(\{N\ge i\}\) is determined by \(W_1,\ldots,W_{i-1}\).
Since \(\E N<\infty\), absolute summability permits termwise expectation:
\[
\E S_N
=\sum_{i\ge1}\E\!\left[u^\top(W_i-J)u\,\mathbf1_{\{N\ge i\}}\right]=0.
\]
At the crossing, \(K_N=n_0+O_h\) exactly. Rearranging
\eqref{app:Kexpansion}, taking expectations, and using
\eqref{app:stoppedmoments} and \(\E|R_N|\to0\) gives
\[
\E N=n_0+\E O_h+
\E\left[\frac{Q_N}{qN}-\frac{S_N^2}{q^2N}\right]-\E R_N
=n_0+\frac{\delta}{q}-\frac{\zeta^2}{q^2}+\E O_h+o(1).
\]
It remains to bound the overshoot. On \(E_h\), both information matrices at
\(N-1\) and \(N\) are positive definite with uniformly bounded condition
numbers for sufficiently small \(h\). The derivative of
\(g(M)=q/[z^\top M^{-1}z]\) is
\[
Dg(M)[H]=q\frac{z^\top M^{-1}HM^{-1}z}{(z^\top M^{-1}z)^2}.
\]
Its norm is uniformly bounded on the segment from \(I_{N-1}\) to \(I_N\):
the scale of \(M\) cancels, and its condition number is bounded there.
As \(\|W_N\|\le L\),
\(0\le O_h\le K_N-K_{N-1}\le C\) on \(E_h\).
On \(E_h^c\), \(O_h\le K_N\le CN\), whose expectation on that event is
\(o(1)\) by the preceding tail bound. Thus \(\E O_h=O(1)\).
A constant \(\nu\) can replace \(\E O_h\) only under the additional
convergence assumption stated in the proposition.

\paragraph{Proof of Proposition~\ref{app:propbeta}.}
Positive-semidefinite information increments make \(K_n\) nondecreasing,
including the convention \(K_n=0\) before positive definiteness. Therefore
\[
\{a_\alpha\sqrt{V_n}\le h\}=\{K_n\ge n_0\}=\{N_h\le n\}
\]
for every integer \(n\). The exact smallest protected size is consequently
the stated integer quantile. If \(\kappa>0\), Theorem~\ref{app:thmclt}
has a continuous, strictly increasing limiting distribution, so its
\((1-\beta)\)-quantiles converge to \(\kappa z_{1-\beta}\). Rescaling yields
\eqref{app:protected}. Rounding the normal approximation changes it by at
most one, which is \(o(\sqrt{n_0})\); the same distributional limit gives
protection probability tending to \(1-\beta\), not an exact finite-sample
inequality.

\section{Additional numerical results}
\label{app:numerical}

\subsection{Estimator sensitivity}

Table~\ref{app:tabestimators} gives the detailed MLE, ridge, Firth, and FLIC
comparison summarized in the main paper.  All methods return finite fits well
before their decisions become certifiable.

\begin{table}[htbp]
\centering
\caption{Estimator sensitivity over 1,000 replications. Coefficient entries are
median Euclidean norms; the final column gives the range across estimators.}
\label{app:tabestimators}
\begin{tabular}{rrrrrrr}
\toprule
\(n\) & events & MLE & ridge & Firth & FLIC & certified / coverage\\
\midrule
500    & 9.9    & 4.88 & 4.64 & 4.54 & 4.78 & 0\% / 97.9--99.4\%\\
2,000  & 39.9   & 4.58 & 4.53 & 4.51 & 4.57 & 0\% / 99.3--99.8\%\\
8,000  & 160.1  & 4.53 & 4.52 & 4.51 & 4.53 & 0\% / 99.7--99.8\%\\
64,000 & 1283.3 & 4.51 & 4.51 & 4.51 & 4.51 & 99.3--99.4\% / 99.6--99.7\%\\
\bottomrule
\end{tabular}
\end{table}

\subsection{Oracle second-order diagnostic}

We ran a separate, fully specified experiment for the oracle precision limit
in Appendix~\ref{app:secondorder}. Let \(X\sim\mathrm{Uniform}[-1,1]\),
\(p(X)=\expit\{\log(0.02/0.98)+0.8X\}\), and \(z=(1,1)^\top\).
We used true information weights, \(a_\alpha=1.96\), every integer look,
and \(h=a_\alpha\sqrt{q/n_0}\) for
\(n_0\in\{375,750,1500,3000,6000\}\).
There were 1,000 independent streams, paired across the five settings, with
no censoring; the seed was 20260828.
Gauss--Legendre quadrature with 128 nodes, checked against 256 nodes, gave
\(q=136.013\), \(\kappa^2=1.888\), and
\(\delta/q-\zeta^2/q^2=0.348\).

Table~\ref{app:tabsecondorder} reports stopping-time means, scaled variances,
and protection quantiles. To examine the mean correction without the leading
centered-sum fluctuation, Table~\ref{app:tabcorrection} also reports
\(D_h=N_h-n_0-O_h+S_{N_h}/q\), using the notation of
Appendix~\ref{app:prooforacle}. Since \(\E S_{N_h}=0\),
\(\E D_h\to\delta/q-\zeta^2/q^2\).
The normal budgets achieved empirical precision probabilities from 94.1\% to
95.5\%; finite-sample and Monte Carlo discrepancies remain.
Figure~\ref{app:figsecondorder} displays the normal approximation and centered
correction. This is a diagnostic of the stated oracle expansions, not a proof
of them or evidence for sequential-band validity or exact finite-sample
protection. Code, all 5,000 replication-setting rows, and generated tables and
figures accompany the submission.

\begin{table}[htbp]
\centering
\caption{Oracle precision diagnostic: 1,000 paired streams per setting. Parentheses are Monte Carlo standard errors of the mean. The final column is the rounded asymptotic 95\% protection budget.}
\label{app:tabsecondorder}
\begin{tabular}{r r r r r}
\toprule
\(n_0\) & mean \(N_h\) (SE) & \(\widehat{\Var}(N_h)/n_0\) & empirical \(Q_{.95}\) & normal budget\\
\midrule
375 & 375.93 (0.83) & 1.817 & 419 & 419\\
750 & 750.70 (1.15) & 1.763 & 811 & 812\\
1,500 & 1,501.18 (1.63) & 1.781 & 1,589 & 1,588\\
3,000 & 3,003.11 (2.29) & 1.754 & 3,128 & 3,124\\
6,000 & 6,004.56 (3.42) & 1.953 & 6,186 & 6,176\\
\bottomrule
\end{tabular}
\end{table}

\begin{table}[htbp]
\centering
\caption{Stopped-moment diagnostic. The centered quantity is \(D_h=N_h-n_0-O_h+S_{N_h}/q\); its limiting mean is 0.348. Parentheses are Monte Carlo standard errors.
This diagnostic does not verify sequential-band validity.}
\label{app:tabcorrection}
\begin{tabular}{r r r r}
\toprule
\(n_0\) & mean \(D_h\) (SE) & mean \(O_h\) (SE) & mean \(R_{N_h}\)\\
\midrule
375 & 0.326 (0.015) & 1.404 (0.041) & -0.0019\\
750 & 0.338 (0.015) & 1.407 (0.039) & -0.0011\\
1,500 & 0.352 (0.016) & 1.444 (0.042) & -0.0010\\
3,000 & 0.360 (0.016) & 1.386 (0.040) & -0.0006\\
6,000 & 0.390 (0.017) & 1.490 (0.041) & -0.0002\\
\bottomrule
\end{tabular}
\end{table}

\begin{figure}[htbp]
\centering
\includegraphics[width=0.95\textwidth]{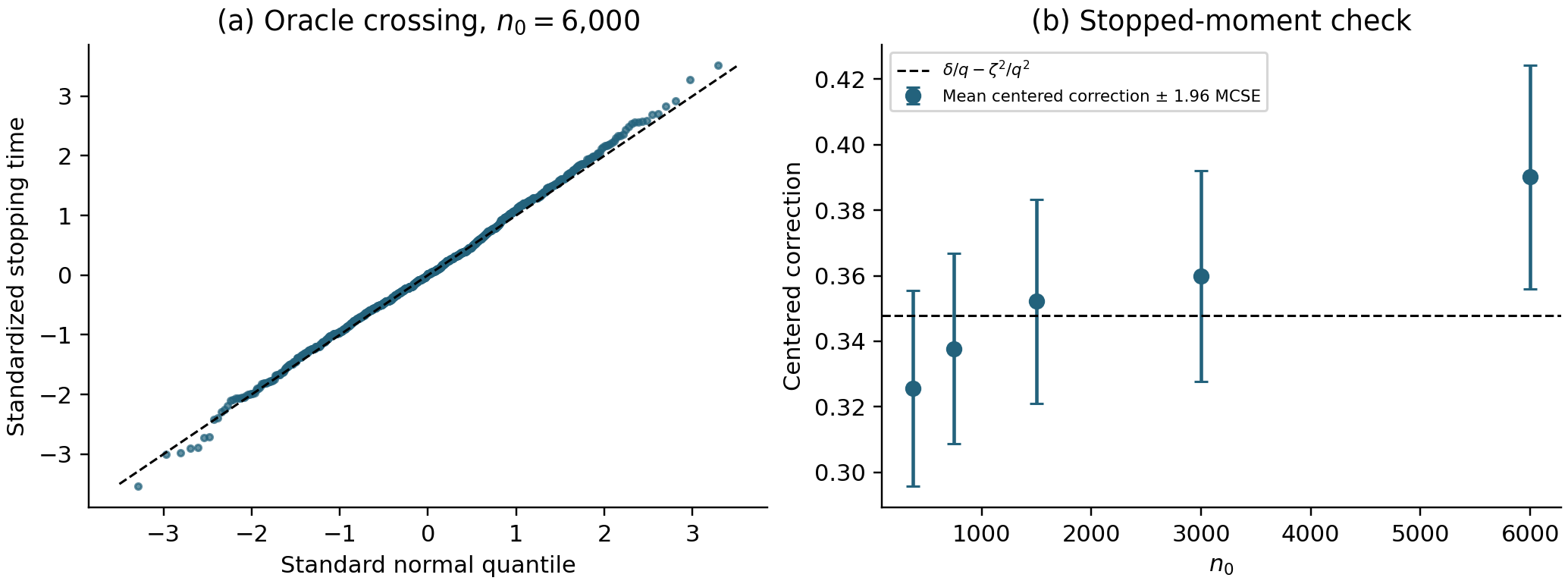}
\caption{Oracle precision diagnostic using 1,000 streams. Panel (a) compares
\((N_h-n_0)/(\kappa\sqrt{n_0})\) with standard normal quantiles at
\(n_0=6000\). Panel (b) shows the mean centered correction with
\(\pm1.96\) Monte Carlo standard errors and its asymptotic value (dashed line).
The five settings share streams; error bars are pointwise.}
\label{app:figsecondorder}
\end{figure}

\clearpage
\section{Full NCHS sensitivity and threshold performance}
\label{app:nchs}

Table~\ref{app:tabgrid} and Figure~\ref{app:figsensitivity} report all 27
prespecified threshold--tolerance cells.  The decision-targeted rule is
right-censored in every cell.  Table~\ref{app:tabthreshold} gives the detailed
2014 threshold-performance comparison.

\begin{table}[tbp]
\centering
\scriptsize
\caption{Sensitivity across the 27 prespecified threshold--tolerance settings.}
\label{app:tabgrid}
\resizebox{\textwidth}{!}{%
\begin{tabular}{rrrrrrrrrr}
\toprule
$\tau$ & $d$ & $\varepsilon$ & stop$_{\rm all}$ & med. $T_R^{\rm all}$ & stop$_\dagger$ & med. $T_R^\dagger$ & $A_\dagger(H)$ & $P_\dagger(H)$ & $\Psi_\dagger(H)$ \\
\midrule
0.01 & 0.01 & 0.05 & 0.000 & 300000 & 0.000 & 300000 & 0.618 & 0.094 & 12.365 \\
0.01 & 0.01 & 0.1 & 1.000 & 141345 & 0.000 & 300000 & 0.618 & 0.094 & 9.433 \\
0.01 & 0.01 & 0.2 & 1.000 & 141345 & 0.000 & 300000 & 0.618 & 0.094 & 9.433 \\
0.01 & 0.025 & 0.05 & 0.000 & 300000 & 0.000 & 300000 & 0.618 & 0.094 & 12.365 \\
0.01 & 0.025 & 0.1 & 1.000 & 141345 & 0.000 & 300000 & 0.618 & 0.094 & 6.182 \\
0.01 & 0.025 & 0.2 & 1.000 & 45711 & 0.000 & 300000 & 0.618 & 0.094 & 3.773 \\
0.01 & 0.05 & 0.05 & 0.000 & 300000 & 0.000 & 300000 & 0.618 & 0.094 & 12.365 \\
0.01 & 0.05 & 0.1 & 1.000 & 141345 & 0.000 & 300000 & 0.618 & 0.094 & 6.182 \\
0.01 & 0.05 & 0.2 & 1.000 & 45711 & 0.000 & 300000 & 0.618 & 0.094 & 3.091 \\
0.025 & 0.01 & 0.05 & 1.000 & 141345 & 0.000 & 300000 & 0.492 & 0.188 & 18.812 \\
0.025 & 0.01 & 0.1 & 1.000 & 141345 & 0.000 & 300000 & 0.492 & 0.188 & 18.812 \\
0.025 & 0.01 & 0.2 & 1.000 & 141345 & 0.000 & 300000 & 0.492 & 0.188 & 18.812 \\
0.025 & 0.025 & 0.05 & 1.000 & 66595 & 0.000 & 300000 & 0.492 & 0.188 & 9.848 \\
0.025 & 0.025 & 0.1 & 1.000 & 31376 & 0.000 & 300000 & 0.492 & 0.188 & 7.525 \\
0.025 & 0.025 & 0.2 & 1.000 & 31376 & 0.000 & 300000 & 0.492 & 0.188 & 7.525 \\
0.025 & 0.05 & 0.05 & 1.000 & 66595 & 0.000 & 300000 & 0.492 & 0.188 & 9.848 \\
0.025 & 0.05 & 0.1 & 1.000 & 31376 & 0.000 & 300000 & 0.492 & 0.188 & 4.924 \\
0.025 & 0.05 & 0.2 & 1.000 & 14783 & 0.000 & 300000 & 0.492 & 0.188 & 3.762 \\
0.05 & 0.01 & 0.05 & 1.000 & 141345 & 0.000 & 300000 & 0.410 & 0.213 & 21.268 \\
0.05 & 0.01 & 0.1 & 1.000 & 141345 & 0.000 & 300000 & 0.410 & 0.213 & 21.268 \\
0.05 & 0.01 & 0.2 & 1.000 & 141345 & 0.000 & 300000 & 0.410 & 0.213 & 21.268 \\
0.05 & 0.025 & 0.05 & 1.000 & 31376 & 0.000 & 300000 & 0.410 & 0.213 & 8.507 \\
0.05 & 0.025 & 0.1 & 1.000 & 31376 & 0.000 & 300000 & 0.410 & 0.213 & 8.507 \\
0.05 & 0.025 & 0.2 & 1.000 & 31376 & 0.000 & 300000 & 0.410 & 0.213 & 8.507 \\
0.05 & 0.05 & 0.05 & 1.000 & 31376 & 0.000 & 300000 & 0.410 & 0.213 & 8.205 \\
0.05 & 0.05 & 0.1 & 1.000 & 14783 & 0.000 & 300000 & 0.410 & 0.213 & 4.254 \\
0.05 & 0.05 & 0.2 & 1.000 & 14783 & 0.000 & 300000 & 0.410 & 0.213 & 4.254 \\
\bottomrule
\end{tabular}}
\begin{minipage}{0.98\linewidth}\footnotesize $H=300{,}000$; terminal quantities include censored replications. Certification requires $\Psi\le1$.\end{minipage}
\end{table}

\begin{figure}[htbp]
\centering
\includegraphics[width=0.98\textwidth]{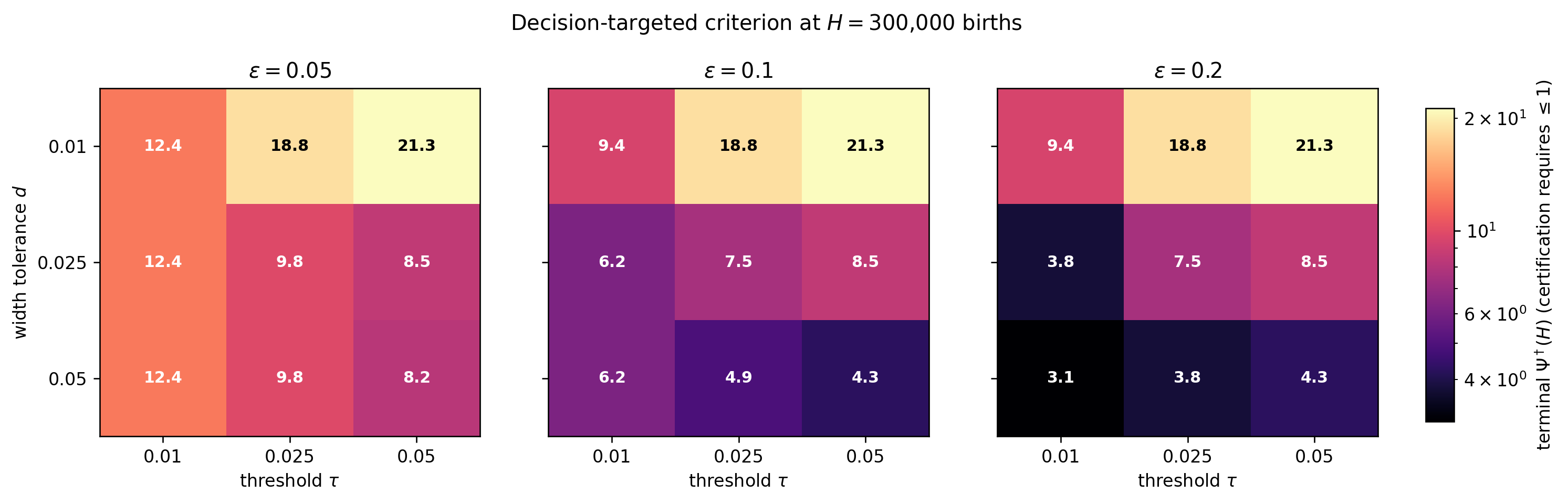}
\caption{Terminal decision-targeted criterion over all 27 NCHS sensitivity
settings. Every value exceeds one, so $T^\dagger$ is right-censored in every
cell.}
\label{app:figsensitivity}
\end{figure}

\begin{table}[htbp]
\centering
\small
\caption{External 2014 threshold performance at $\tau=0.05$.}
\label{app:tabthreshold}
\resizebox{\textwidth}{!}{%
\begin{tabular}{lrrrrrrr}
\toprule
fit & sensitivity & specificity & PPV & NPV & referral fraction & net benefit & NB all \\
\midrule
$T^{\mathrm{all}}$ & 0.501 & 0.989 & 0.197 & 0.997 & 0.014 & 0.0022 & -0.0467 \\
$T^\dagger$ horizon (censored) & 0.499 & 0.989 & 0.197 & 0.997 & 0.014 & 0.0022 & -0.0467 \\
full development sample & 0.499 & 0.989 & 0.197 & 0.997 & 0.014 & 0.0022 & -0.0467 \\
\bottomrule
\end{tabular}}
\begin{minipage}{0.98\linewidth}\footnotesize Treat-none net benefit is zero. All quantities use the same disjoint 2014 outcome-validation sample.\end{minipage}
\end{table}

\clearpage
\section{MODY-calibrated plasmode stress test}
\label{app:mody}

This plasmode considers population-scale screening for maturity-onset diabetes
of the young (MODY), motivated by an All of Us genotype-first study of
$N=374{,}973$ participants \cite{hasebe2026,magee2025}.  No individual-level
All of Us records enter the analysis.  The simulation uses twelve standardized
clinical predictors, full-model AUC approximately 0.91, a genetic-testing
threshold $\tau=0.10$, and tolerances
\((d,\varepsilon)=(0.05,0.10)\).  Two prevalences bracket the motivating
settings: $\pi=0.02$ among people with diabetes and $\pi=0.003$ in a
general genotype-first cohort.  Results average 1,000 accrual orderings.

\begin{table}[htbp]
\centering
\caption{Sequential certification on MODY-calibrated data (plasmode; twelve
predictors, AUC approximately 0.91, 1,000 replications).}
\label{app:tabmody}
\begin{tabular}{lrr}
\toprule
Quantity & Among diabetics ($\pi=0.02$) & General ($\pi=0.003$)\\
\midrule
Whole-population baseline $T^{\mathrm{all}}$ (mean) & 9,000 & 4,200\\
\(\beta\)-protection $Q_{0.95}(T^{\mathrm{all}})$ & 11,000 & ---\\
genuine cases at $T^{\mathrm{all}}$ & approximately 179 & approximately 12\\
certified-decision agreement & 99.97\% & 99.99\%\\
High-risk referral certified & 100\% ($T\approx680$) & 50\%\\
Borderline profile not certified & 53\% & 43\%\\
\bottomrule
\end{tabular}
\end{table}

Among people with diabetes, the high-risk phenotype is certified after about
680 patients in every ordering, while the borderline profile is not certified in
53\% of orderings.  In the general cohort, positive referral for the high-risk
profile is certifiable in only about half the orderings even at 400,000 participants,
and the borderline profile is not certified in 43\% of orderings.  The study therefore supports enrichment
before certification.  It is a prevalence-and-enrichment stress test, not a
second observed-data validation.

\section*{Reproducibility}

All simulation and NCHS analyses use 1{,}000 replications per cell. For the
NCHS application and the oracle precision diagnostic, the accompanying materials
include the generating scripts, random-number seeds, monitoring grids,
configuration, unit tests, and replication-level outputs, and document the steps
needed to reproduce both analyses from the official public data sources.  For the
MODY-calibrated plasmode, the generating script and retained run summary are
included.  For Studies~1--4, the ancillary reproducibility archive includes the frozen numerical summaries
used in the manuscript, the retained legacy scripts, and a clean reference
implementation of the documented certification mechanics and configurations.
Because the exact replication-level files and seeds from the final historical
Study~1--4 run were not preserved, exact byte-for-byte
regeneration of those displayed Monte Carlo summaries is not claimed.

\section*{Author contributions}
Hui-Mean Foo: software/programming, literature review, and writing--original draft and review/editing.
\par
Yuan-chin Ivan Chang: conceptualization, methodology, theory, software,
formal analysis, data curation, visualization, and writing.

\section*{Funding information}
The authors report no dedicated external funding for this study.

\section*{Conflict of interest}
The authors declare no conflict of interest.

\section*{Data availability statement}
The observed-data analysis uses the publicly available US National Center for
Health Statistics Cohort Linked Birth/Infant Death public-use files.  The
manuscript and accompanying materials document the cohort definitions,
exclusions, sampling seeds, and data-preparation steps needed to reproduce the
analytic samples.  The
MODY-calibrated analysis is a simulated plasmode and does not redistribute or
analyze individual-level controlled-access genomic records.

\section*{Code availability statement}
Code and simulation materials supporting the analyses are included with this
preprint as an ancillary archive under the arXiv \texttt{anc/} directory.  The NCHS and oracle components include saved
configurations, seeds, tests, and replication-level outputs; the Study~1--4 materials
include the frozen reported summaries and a clean reference implementation of the
documented certification mechanics and configurations, with the historical provenance
limitation described above.  The same materials will also be archived in a public repository upon journal acceptance.

\section*{Use of artificial intelligence tools}
During preparation of the final submission, OpenAI ChatGPT (GPT-5.6 Sol; accessed
28 August--1 September 2026) was used to assist with language editing, bibliographic
cross-checking, consistency review, and review of code/documentation.  It was not
used to generate the research data reported in the manuscript or to make autonomous
statistical or clinical interpretations.  The authors reviewed the suggested changes,
verified the mathematical statements, references, code, and reported results, and
take full responsibility for the content.

\section*{Ethics statement}
The analysis used publicly available, deidentified NCHS public-use files.  No
identifiable individual-level information was accessed, and no separate
data-use approval was required.

\section*{Acknowledgments}
The infant-mortality analysis uses the public-use NCHS Cohort Linked
Birth/Infant Death file, which requires no separate data-use approval.


\end{document}